\documentclass[11pt]{article}
\usepackage[utf8]{inputenc} 
\usepackage[T1]{fontenc} 
\usepackage{amsthm}
\usepackage{mathtools}

\usepackage{graphicx} 
\usepackage{array} 

\usepackage{amsmath, amssymb, amsfonts, dsfont, verbatim}
\usepackage{hyphenat,epsfig,subcaption,multirow}
\usepackage{nicefrac}
\usepackage{paralist}
\usepackage{enumitem}
\usepackage{adjustbox}

\usepackage[font=footnotesize,labelfont=bf]{caption}

\usepackage[usenames,dvipsnames]{xcolor}
\usepackage[ruled]{algorithm2e}

\DeclareFontFamily{U}{mathx}{\hyphenchar\font45}
\DeclareFontShape{U}{mathx}{m}{n}{
      <5> <6> <7> <8> <9> <10>
      <10.95> <12> <14.4> <17.28> <20.74> <24.88>
      mathx10
      }{}
\DeclareSymbolFont{mathx}{U}{mathx}{m}{n}
\DeclareMathSymbol{\bigtimes}{1}{mathx}{"91}

\usepackage{tcolorbox}
\tcbuselibrary{skins,breakable}
\tcbset{enhanced jigsaw}

\usepackage[normalem]{ulem}
\usepackage[compact]{titlesec}

\definecolor{DarkRed}{rgb}{0.5,0.1,0.1}
\definecolor{DarkBlue}{rgb}{0.1,0.1,0.5}
\definecolor{RURed}{rgb}{0.8,0.1,0.1}

\usepackage{nameref}
\definecolor{ForestGreen}{rgb}{0.1333,0.5451,0.1333}
\definecolor{Red}{rgb}{0.9,0,0}
\usepackage[linktocpage=true,
	pagebackref=true,colorlinks,
	linkcolor=DarkRed,citecolor=ForestGreen,
	bookmarks,bookmarksopen,bookmarksnumbered]
	{hyperref}
\usepackage[noabbrev,nameinlink]{cleveref}
\crefname{property}{property}{Property}
\creflabelformat{property}{(#1)#2#3}
\crefname{equation}{eq}{Eq}
\creflabelformat{equation}{(#1)#2#3}

\usepackage{bm}
\usepackage{url}
\usepackage{xspace}
\usepackage[mathscr]{euscript}
\usepackage{colortbl}

\usepackage{tikz}
\usetikzlibrary{arrows}
\usetikzlibrary{arrows.meta}
\usetikzlibrary{shapes}
\usetikzlibrary{backgrounds}
\usetikzlibrary{positioning}
\usetikzlibrary{decorations.markings}
\usetikzlibrary{patterns}
\usetikzlibrary{calc}
\usetikzlibrary{fit}
\usetikzlibrary{snakes}
\tikzset{vertex/.style={circle, black, fill=Yellow, line width=1pt, draw, minimum width=8pt, minimum height=8pt, inner sep=0pt}}
	
\usepackage{mdframed}
\usepackage[noend]{algpseudocode}
\makeatletter
\def\BState{\State\hskip-\ALG@thistlm}
\makeatother

\usepackage{cite}
\usepackage{enumitem}
\usepackage{todonotes}
\usepackage[margin=1in]{geometry}
\usepackage{thmtools} 
\usepackage{thm-restate}

\newtheorem{theorem}{Theorem}
\newtheorem{lemma}{Lemma}[section]
\newtheorem{proposition}[lemma]{Proposition}

\newtheorem*{claim*}{Claim}
\newtheorem*{theorem*}{Theorem}
\newtheorem*{proposition*}{Proposition}
\newtheorem*{lemma*}{Lemma}
\newtheorem*{problem*}{Problem}

\crefname{lemma}{Lemma}{Lemmas}
\crefname{appendix}{Appendix}{Appendices}
\crefname{proposition}{Proposition}{Propositions}
\crefname{claim}{Claim}{Claims}
\crefname{example}{Example}{Examples}
\crefname{figure}{Figure}{Figures}

\theoremstyle{plain}
\newtheorem{mdresult}{Result}

\newtheorem{observation}[lemma]{Observation}
 \theoremstyle{definition}
 \newtheorem{definition}[lemma]{Definition}

\newtheorem{remark}[lemma]{Remark}
\newtheorem*{remark*}{Remark}

\newtheoremstyle{restate}{}{}{\itshape}{}{\bfseries}{~(restated).}{.5em}{\thmnote{#3}}
\theoremstyle{restate}

\theoremstyle{definition}
\newtheorem{mdalg}{algorithm}

\newtheorem{mddist}{Distribution}

\allowdisplaybreaks

\DeclareMathOperator*{\argmax}{arg\,max}

\renewcommand{\qed}{\nobreak \ifvmode \relax \else
      \ifdim\lastskip<1.5em \hskip-\lastskip
      \hskip1.5em plus0em minus0.5em \fi \nobreak
      \vrule height0.75em width0.5em depth0.25em\fi}

\newcommand*\samethanks[1][\value{footnote}]{\footnotemark[#1]}

\makeatletter
\newcommand{\thickhline}{%
    \noalign {\ifnum 0=`}\fi \hrule height 1.2pt
    \futurelet \reserved@a \@xhline
}

\let\originalleft\left
\let\originalright\right
\renewcommand{\left}{\mathopen{}\mathclose\bgroup\originalleft}
\renewcommand{\right}{\aftergroup\egroup\originalright}

\newcommand{\Ot}{\ensuremath{\widetilde{O}}}

\newcommand{\eps}{\ensuremath{\varepsilon}}
\newcommand{\Paren}[1]{\Big(#1\Big)}

\newcommand{\bracket}[1]{\left[#1\right]}
\newcommand{\paren}[1]{\ensuremath{\left(#1\right)}\xspace}
\newcommand{\card}[1]{\left\vert{#1}\right\vert}

\newcommand{\norm}[1]{\ensuremath{\|#1\|}}
\newcommand{\Lap}[1]{{\ensuremath{\textnormal{\textsf{Lap}}(#1)}\xspace}}
\newcommand{\abs}[1]{\ensuremath{\left|#1\right|}}
\newcommand{\bx}{\mathbf{x}}

\newcommand{\prob}[1]{\Pr\bracket{#1}} 
\newcommand{\expect}[1]{\Exp\bracket{#1}}

\newcommand{\set}[1]{\ensuremath{\left\{ #1 \right\}}}

\newcommand{\polylog}{\textnormal{polylog}\xspace}

\DeclareMathOperator*{\Exp}{\ensuremath{{\mathbb{E}}}}
\DeclareMathOperator*{\Prob}{\ensuremath{\textnormal{Pr}}}
\renewcommand{\Pr}{\Prob}

\newenvironment{tbox}{\begin{tcolorbox}[
		enlarge top by=5pt,
		enlarge bottom by=5pt,
		 breakable,
		 boxsep=2pt,
                  left=5pt,
                  right=7pt,
                  top=10pt,
                  arc=0pt,
                  boxrule=1pt,toprule=1pt,
                  colback=white
                  ]
	}
{\end{tcolorbox}}

\newcommand{\II}{\ensuremath{\mathbb{I}}}

\newcommand{\mireal}[1][]{
  \ifx\relax#1\relax%
    \II(\mione \,; \mitwo)%
  \else%
    \II(\mione \,; \mitwo\mid #1)%
  \fi
}

\newcommand{\dist}{\mathsf{dist}}

\newcommand{\cP}{\ensuremath{\mathcal{P}}}

\newcommand{\cA}{\ensuremath{\mathcal{A}}}

\newcommand{\algname}[1]{\textbf{\textsc{#1}}}

\makeatletter
\DeclareFontFamily{U}{tipa}{}
\DeclareFontShape{U}{tipa}{m}{n}{<->tipa10}{}
\newcommand{\arc@char}{{\usefont{U}{tipa}{m}{n}\symbol{62}}}%

\newcommand{\arc}[1]{\mathpalette\arc@arc{#1}}

\newcommand{\arc@arc}[2]{%
  \sbox0{$\m@th#1#2$}%
  \vbox{
    \hbox{\resizebox{\wd0}{\height}{\arc@char}}
    \nointerlineskip
    \box0
  }%
}
\makeatother

\newcommand{\inorm}[1]{\ensuremath{\|#1\|_{\infty}}}
\newcommand{\onenorm}[1]{\ensuremath{\|#1\|_{1}}}
\newcommand{\vsens}{\ensuremath{\arc{\mathtt{v}}}}
\newcommand{\esens}{\ensuremath{\arc{\mathtt{e}}}}

\newcommand{\edgespan}{\ensuremath{\mathsf{span}}}
\newcommand{\inport}{\ensuremath{\mathsf{in}}}
\newcommand{\outport}{\ensuremath{\mathsf{out}}}

\newcommand{\cluster}{\ensuremath{\textrm{Cluster}}\xspace}
\newcommand{\bunch}{\ensuremath{\textrm{Bunch}}\xspace}
\newcommand{\ball}{\ensuremath{\textrm{Ball}}\xspace}

\newcommand{\pathshortcut}{\textsc{Path-Hopset}\xspace}
\newcommand{\treeshortcut}{\textsc{Tree-Hopset}\xspace}
\newcommand{\greedyshortcut}{\textsc{Greedy-Hopset}\xspace}

\newcommand{\ezshortcut}{\textsc{Undirected-Shortcut-Set}\xspace}
\newcommand{\folkloreshortcut}{\textsc{Folklore-Hopset}\xspace}

\newcommand{\directedreach}{\textsc{Greedy-Di-Shortcut-Set}\xspace}
\newcommand{\puredpalg}{\ensuremath{\textnormal{\textsc{Hopset-ASRQ}}}\xspace}

\newcommand{\tzspan}{\textsc{Apx-Undirected-Hopset}\xspace}
\newcommand{\tzemu}{\textsc{TZ-Emulator}\xspace}

\title{Low Sensitivity Hopsets}

\author{Vikrant Ashvinkumar\thanks{Rutgers University, \texttt{\{va264,aaron.bernstein,jg1555,cd751\}@rutgers.edu} Ashvinkumar, Deng and Gao would like to acknowledge funding through NSF IIS-2229876, DMS-2220271, DMS-2311064, CCF-2208663, CCF-2118953. Bernstein would like to acknowledge funding through Sloan Fellowship, Google Research Award, and NSF Grant 1942010.} \and Aaron Bernstein\samethanks \and Chengyuan Deng\samethanks \and Jie Gao\samethanks \and Nicole Wein\thanks{University of Michigan, \texttt{nswein@umich.edu}. This work was done while the author was at DIMACS, Rutgers University,  supported by a grant to DIMACS from the Simons
Foundation (820931). }}

\date{}

\begin{document}

\maketitle

\begin{abstract}
    Given a weighted graph $G=(V,E,w)$, a $(\beta, \eps)$-hopset $H$ is an edge set such that for any $s,t \in V$, where $s$ can reach $t$ in $G$, there is a path from $s$ to $t$ in $G \cup H$ which uses at most $\beta$ hops whose length is in the range $[\dist_G(s,t), (1+\eps)\dist_G(s,t)]$.
    We break away from the traditional question that asks for a hopset $H$ that achieves small $|H|$ and small diameter $\beta$ and instead study the \emph{sensitivity} of $H$, a new quality measure.
    The sensitivity of a vertex (or edge) given a hopset $H$ is, informally, the number of times a single hop in $G \cup H$ bypasses it; a bit more formally, assuming shortest paths in $G$ are unique, it is the number of hopset edges $(s,t) \in H$ such that the vertex (or edge) is contained in the unique $st$-path in $G$ having length exactly $\dist_G(s,t)$.
    The sensitivity associated with $H$ is then the maximum sensitivity over all vertices (or edges).
    The highlights of our results are: 
    \begin{itemize}
        \item A construction for $(\Ot(\sqrt{n}), 0)$-hopsets on undirected graphs with $O(\log n)$ sensitivity, complemented with a lower bound showing that $\Ot(\sqrt{n})$ is tight up to polylogarithmic factors for any construction with polylogarithmic sensitivity.
        \item A construction for $(n^{o(1)}, \eps)$-hopsets on undirected graphs with $n^{o(1)}$ sensitivity for any $\eps > 0$ that is at least inverse polylogarithmic, complemented with a lower bound on the tradeoff between $\beta, \eps$, and the sensitivity.
        \item We define a notion of sensitivity for $\beta$-shortcut sets (which are the reachability analogues of hopsets) and give a construction for $\Ot(\sqrt{n})$-shortcut sets on \emph{directed} graphs with $O(\log n)$ sensitivity, complemented with a lower bound showing that $\beta = \widetilde{\Omega}(n^{1/3})$ for any construction with polylogarithmic sensitivity.
    \end{itemize}

    We believe hopset sensitivity is a natural measure in and of itself, and could potentially find use in a diverse range of contexts.
    More concretely, the notion of hopset sensitivity is also directly motivated by the Differentially Private All Sets Range Queries problem \hyperlink{abs}{[Deng et al. WADS 23]}.
    Our result for $O(\log n)$ sensitivity $(\Ot(\sqrt{n}), 0)$-hopsets on undirected graphs immediately improves the current best-known upper bound on  utility from $\Ot(n^{1/3})$ to $\Ot(n^{1/4})$ in the pure-DP setting, which is tight up to polylogarithmic factors.
\end{abstract}

\pagenumbering{roman}
\clearpage
\tableofcontents
\clearpage
\pagenumbering{arabic}
\setcounter{page}{1}

\section{Introduction}
\label{sec:intro}
Many fundamental graph-theoretic problems involve the computation of reachability and shortest path distances.
In a variety of computation models (e.g. parallel, distributed, streaming, dynamic), the computation of shortest paths of most state-of-the-art algorithms scales with the \emph{hop-diameter}~\cite{klein1997randomized,henzinger2014sublinear,henzinger2015improved,forster2018faster,jambulapati2019parallel,fineman2018nearly,gutenberg2020decremental,bernstein2020near,cao2020efficient,cao2021brief,karczmarz2021deterministic}.
The hop-diameter of a graph refers to the maximum hop-length over reachable pairs of vertices, where the hop-length from vertex $s$ to vertex $t$ is the minimum number of edges in a shortest $st$-path.
Motivated by the goal of reducing hop-diameter, Thorup introduced the notion of a \emph{shortcut set}~\cite{thorup1993shortcutting}, which is a set of edges $H$ added to a given graph $G$ such that: $G \cup H$ has the same reachability structure as $G$ (i.e. for all vertices $s$ and $t$, $s$ can reach $t$ in $G \cup H$ if and only if $s$ can reach $t$ in $G$)
and $G \cup H$ has small hop-diameter.
The shortcut set was later generalized to the \emph{hopset} by Cohen~\cite{cohen2000polylog} (also informally by~\cite{klein1997randomized, shi1999time}) which preserves (weighted) shortest distances in addition to reachability.
Formal definitions are given below.

\begin{definition}[\textit{Shortcut Set}]
    \label{def:shortcut-set}
    Given a graph $G = (V,E)$, a \emph{$\beta$-shortcut set} of $G$ is a set of edges $H \subseteq V \times V$ such that: (1) Every edge $(s,t) \in H$ is in the transitive closure of $G$. (2) For every vertex pair $(s,t)$ in the transitive closure of $G$, there is an $st$-path in $G\cup H$ with at most $\beta$ edges.
\end{definition}

\begin{definition}[\textit{Hopset}]
    \label{def:hopset}
    \normalfont Given a weighted graph $G = (V,E,w)$, a \emph{$(\beta,\eps)$-hopset} of $G$ is a set of weighted edges $H \subseteq V \times V$ such that: (1) Every edge $(s,t) \in H$ has a weight of $\dist_{G}(s,t)$, where $\dist_{G}(s,t)$ stands for the shortest distance between $s,t$ in graph $G$. (2) For every vertex pair $(s,t)$ in the transitive closure of $G$, there is an $st$-path $P_{st}$ in $G \cup H$ with at most $\beta$ edges, and the weight of $P_{st}$ is at most $(1+\eps)\dist_{G}(s,t)$.
\end{definition}

A $(\beta,0)$-hopset is sometimes called an exact hopset, while $(\beta,\eps)$-hopsets for $\eps > 0$ are called approximate hopsets.
For both shortcut sets and hopsets, a natural measure of their \emph{cost} is their size (i.e. the number of edges added).
There has been a rich literature studying the tradeoff between the size of a shortcut/hopset and the hop-diameter for both directed and undirected graphs~\cite{ullman1990high,hesse2003directed,huang2019thorup,elkin2019hopsets,lu2022better,kogan2022new,kogan2022beating,kogan2022having,shabat2021unified,ben2020new,abboud2018hierarchy,bernstein2023closing,bodwin2023folklore,williams2024simpler}.

\paragraph{A New Problem: Low-Sensitivity Hopsets.}
In some settings, measuring the cost of a hopset by its total number of edges provides information that is too coarse and not descriptive enough to capture the problem.
In this work, we instead consider a notion of cost that is more local to individual vertices and edges.
Specifically, we define the notion of the \emph{sensitivity} of a hopset\footnote{To avoid redundancy, we use hopset as a general term for shortcut/approximate/hopset unless specified otherwise.}. 
 
Informally, given a graph $G = (V,E)$ and a hopset $H$, the sensitivity of a vertex $v$ is the number of hopset edges that bypass $v$.
We say a hopset edge $(s,t) \in H$  bypasses $v$ if $v$ is on the unique\footnote{We assume shortest paths are unique now. Later we formally define the conditions for non-unique shortest paths.} $st$-path in $G$ having length exactly $\dist_G(s,t)$.
The vertex sensitivity associated with $H$ is then the maximum sensitivity over all vertices.
The edge sensitivity is defined in a similar way.
We denote the hopset vertex/edge sensitivity by $\inorm{\vsens}$ and $\inorm{\esens}$\footnote{The arc is meant to be evocative of bypassing hopset edges. When they bypass vertices, we draw the arc over the symbol $\mathtt{v}$ and when they bypass edges, we draw the arc over the symbol $\mathtt{e}$.}, and defer the formal definitions to \Cref{subsec:formal-def-hopset}.

We say a hopset $H$ has low sensitivity if it has a vertex/edge sensitivity of $\polylog(n)$. 
We also define an analogous notion of a \emph{low-sensitivity shortcut set}.
Since shortcut sets pertain to reachability instead of shortest paths, and there could be many paths between a vertex pair $s,t$, it is not clear a priori which path should absorb the sensitivity of a shortcut edge from $s$ to $t$.
We allow the algorithm to specify the paths; that is, the input is a graph, and the output is a path between each pair of vertices as well as a shortcut set that has low sensitivity with respect to the chosen paths.

A closely related concept is the \emph{support size} of a hopset, first studied in \cite{elkin2023path} to get recently improved path-reporting distance oracles (PRDOs); see \cite{chechik2024path} for the latest on PRDOs.
The support size of a hopset $H$ is, loosely speaking, the number of edges in $G$ that are bypassed by edges from $H$.
This can be seen as the $\ell_0$ norm of a certain vector whereas, in contrast, hopset sensitivity is the $\ell_\infty$ norm of the same vector (as the notation $\inorm{\esens}$ suggests\footnote{In fact, we also look at the $\ell_1$ norm in the proofs of our lower bounds in \Cref{sec:lb}.}).

At any rate, we investigate tradeoffs between hopset sensitivity and hop-diameter in this work.
We believe hopset sensitivity is a natural measure in and of itself, and could potentially find use in a diverse range of contexts. 
In fact, low-sensitivity hopsets have already been implicitly studied in several prior works~\cite{deng2023differentially, ghazi2022differentially, fan2022distances}.
These prior works come from the area of differential privacy where sensitivity is a crucial concept since it captures the effect of the perturbation of individual data points on the output of an algorithm.
Our new low-sensitivity hopset constructions directly improve the bounds from~\cite{deng2023differentially} for the problem of \emph{Differentially Private Range Queries}.
More details on this concrete application, as well as other applications of a more speculative nature, can be found in \Cref{subsec:app-to-dp}.

\subsection{Our Results}

We study upper and lower bounds for low-sensitivity shortcut sets, exact hopsets, and approximate hopsets, on both undirected and directed graphs.

To give context to our results, we draw parallels to the different regimes for traditional hopsets.
For hopsets with $O(n)$ edges, there are essentially three diameter regimes: (1) directed shortcut/hopsets as well as undirected exact hopsets all provably require polynomial hop-diameter, (2) \emph{approximate} hopsets for undirected graphs require only $n^{o(1)}$ hop-diameter, and (3) undirected shortcut sets trivially achieve diameter 2.
Roughly speaking, these diameter regimes also hold for the low-sensitivity counterparts.

Our results focus on low-sensitivity undirected exact hopsets and directed shortcut sets from regime 1, as well as low-sensitivity approximate undirected hopsets from regime 2.
For all of these, we prove both upper and lower bounds.

An overview of upper and lower bound results of all settings is shown in \Cref{tab:summary}. Observe that the edge sensitivity of a hopset is always no more than the vertex sensitivity (formally proved in \Cref{subsec:sens-transfer}).
Thus, it is always more desirable to have upper bounds for vertex sensitivity and lower bounds for edge sensitivity.
All of our results are of this form.

\begin{table}[h!]
\renewcommand{\arraystretch}{1.8}
\centering
\caption{Our Results for sensitivity and hop-diameter for shortcut sets, exact hopsets, and $(1+\eps)$ hopsets. Regarding the undirected approximate hopset, the upper bound has $\eps > 0$ at least inverse polylogarithmic and, in the lower bound, $k$ is any positive integer and $\Delta$ is any small positive constant.  
} 
\vspace{0.2em}
\label{tab:summary}
\begin{adjustbox}{width=1\textwidth}
\begin{tabular}{c|c|c|c|c} 
\thickhline
 &  & Shortcut Set & Hopset & $(1+\varepsilon)$ Hopset    \\ 
\thickhline

\multirow{5}{*}{\begin{tabular}[c]{@{}c@{}}Undirected \end{tabular}}   & \multirow{3}{*}{\begin{tabular}[c]{@{}c@{}} U. B. \end{tabular}} & $\beta = O(\log^2 n)$ & $\beta = O(\sqrt{n}\log n)$ & $\beta = n^{o(1)}$  \\
 & & $\inorm{\vsens} = O(\log n)$ & $\inorm{\vsens} = O(\log n)$ & $\inorm{\vsens} = n^{o(1)}$ \\
 & & (\cite{fan2022distances}, \Cref{app:warmup-1}) & (\Cref{subsec:undi-hop-upper-proof}) & (\Cref{sec:upper-apx-undi}) \\
\cline{2-5}
 & \multirow{3}{*}{\begin{tabular}[c]{@{}c@{}} L. B. \end{tabular}}  & \multirow{3}{*}{\begin{tabular}[c]{@{}c@{}} - \end{tabular}}  & \multirow{2}{*}{\begin{tabular}[c]{@{}c@{}} $\inorm{\esens} \cdot \beta^2 = \Omega(n)$ \end{tabular}} & $\beta = O_k((1/\eps)^k)$   \\ 
 & & & & $\Rightarrow \inorm{\esens} = \Omega(n^{\frac{1}{2^k - 1} - \Delta})$ \\
 & & & (\Cref{subsec:lower-hop-shortcut})
& (\Cref{subsec:lower-apx-hopset})\\
\hline

\multirow{5}{*}{\begin{tabular}[c]{@{}c@{}}Directed \end{tabular}} & \multirow{3}{*}{\begin{tabular}[c]{@{}c@{}} U. B. \end{tabular}} & $\beta = O(\sqrt{n}\log^3 n)$ & $\beta = O(\sqrt{n}\log n)$ &  $\beta = O(\sqrt{n}\log n)$   \\
 & & $\inorm{\vsens} = O(\log n)$ & $\inorm{\vsens} = O(\sqrt{n}\log n)$ & $\inorm{\vsens} = O(\sqrt{n}\log n)$  \\
 & & (\Cref{subsec:di-short-upper-proof}) & (\cite{deng2023differentially, ghazi2022differentially}, \Cref{app:warmup-2}) & (Implied by directed exact hopset)\\
\cline{2-5}
& \multirow{2}{*}{\begin{tabular}[c]{@{}c@{}} L. B. \end{tabular}} & $\inorm{\esens} \cdot \beta = \Omega(n^{1/3})$ & $\inorm{\esens} \cdot \beta^2 = \Omega(n)$ &  $\inorm{\esens} \cdot \beta = \Omega(n^{1/3})$  \\ 
& & (\Cref{subsec:lower-hop-shortcut}) & (Implied by undirected hopset) & (Implied by directed shortcut set)\\
\thickhline

\end{tabular}
\end{adjustbox}
\end{table}

\paragraph{Undirected Exact Hopsets.} Our first result is for low-sensitivity exact hopsets on undirected graphs.
We show it is possible to achieve hop-diameter $\Ot(\sqrt{n})$ and $O(\log n)$ vertex sensitivity.

\begin{restatable}[Undirected Exact Hopset Upper Bound]{theorem}{greedygood}
\label{thm:greedy-good}
    There exists an algorithm producing an $(O(\sqrt{n}\log{n}), 0)$-hopset with $\inorm{\vsens} = O(\log n)$ over undirected and directed acyclic graphs.
\end{restatable}

We complement our upper bound with a lower bound showing that a hop-diameter of $\Ot(\sqrt{n})$ is \emph{tight} up to polylogarithmic factors, for any hopset with polylogarithmic (vertex- or edge-) sensitivity.

\begin{restatable}[Undirected Exact Hopset Lower Bound]{theorem}{exactlb}
\label{cor:exact-lb}
    Any construction of $(\beta, 0)$-hopsets $H$ must have $\inorm{\esens} \cdot \beta^2 = \Omega(n)$ for a graph $G$ on $n$ vertices.
\end{restatable}

Our lower bound exhibits a smooth trade-off between hop-diameter and sensitivity.
One could ask whether such a trade-off exists for upper bounds as well. As we show (see \Cref{rem:perfect}), such a trade-off would imply the non-existence of certain \emph{perfect path} systems, which is a big open problem in network design. The existence of such perfect path systems would imply, for example, better lower bounds for distance preservers \cite{coppersmith2006sparse}.

\paragraph{Directed Shortcut Sets.} For directed shortcut sets, we prove an upper bound with similar guarantees to our upper bound for undirected exact hopsets.

\begin{restatable}[Directed Shortcut Set Upper Bound]{theorem}{directedreachub}
\label{thm:direted-reach-ub}
    There exists an algorithm producing an $O(\sqrt{n}\log^3 n)$-shortcut set with $\inorm{\vsens} = O(\log n)$ for directed graphs.
\end{restatable}

We complement our upper bound with a lower bound.

\begin{restatable}[Directed Shortcut Set Lower Bound]{theorem}{reachabilitylb}
\label{cor:reachability-lb}
    Any construction of $\beta$-shortcut sets $H$ must have $\inorm{\esens} \cdot \beta = \Omega(n^{1/3})$ for a directed graph $G$ on $n$ vertices.
\end{restatable}

Unlike our bounds for exact undirected hopsets, our results for directed shortcut sets are not tight.
In particular, for polylogarithmic vertex-sensitivity, there is a gap between hop-diameter $\Ot(n^{1/3})$ and $\Ot(\sqrt{n})$, which we leave as an open problem.

\paragraph{Undirected Approximate Hopsets.} 

For undirected approximate hopsets, we prove that one can achieve much better hop-diameter than the polynomial bounds for the previously mentioned problems.
In particular, we prove bounds of $n^{o(1)}$ for both sensitivity and hop-diameter when $\eps > 0$ is at least inverse polylogarithmic.

\begin{restatable}{theorem}{apxhopsetub}
\label{thm:apx-hopset-ub}
    There exists an algorithm which produces a $\paren{O\paren{(k/\eps)^k \log^2 n}, \eps}$-hopset $H$ with $\inorm{\vsens} = O(kn^{1/k}\log^2 n)$ for undirected graphs.
    For any $\eps > 0$ that is at least inverse polylogarithmic, setting $k = \Theta(\sqrt{\log n})$ gives a $(n^{o(1)}, \eps)$-hopset with $\inorm{\vsens} = n^{o(1)}$.
\end{restatable}

We complement our upper bound with a lower bound.

\begin{restatable}{theorem}{apxhopsetlb}
\label{thm:apx-hopset-lb}
    Fix a positive integer $k$ and parameter $\eps > 1/n^{o(1)}$.
    Any construction of $(\beta, \eps)$-hopsets $H$ with $\beta = O\paren{\paren{\frac{1}{2^{k-2}(2k-1)\eps}}^k}$ has $\inorm{\esens} \ge n^{\frac{1}{2^k - 1} - \Delta}$, $\Delta > 0$.
\end{restatable}

This lower bound is an exact analogue of the best known lower bound on the size of approximate hopsets in \cite{abboud2018hierarchy} (we just divide their size bound by $n$ to get our sensitivity bound).
For $\eps > 0$ exactly inverse polylogarithmic, this lower bound says that it is not possible to have sensitivity and hop-diameter be polylogarithmic simultaneously, so we can only hope to improve either the sensitivity or diameter under this regime.
We leave open whether we can get sensitivity and diameter simultaneously polylogarithmic when $\eps$ is larger than polylogarithmic (for example, constant $\eps$).

To briefly address the problems that we do not focus on (recall that we focus on undirected exact hopsets, directed shortcut sets, and approximate undirected hopsets):
Low-sensitivity undirected shortcut sets are simple (but not quite as trivial as in the traditional setting) as noted by prior work and in \Cref{obs:tree-shortcut}.
For directed hopsets, the upper bounds implicit in prior work~\cite{deng2023differentially, ghazi2022differentially} remain the best-known (see \Cref{app:warmup-1,app:warmup-2}), and from the lower bounds side, our results for undirected exact hopsets and directed shortcut sets immediately carry over to exact and approximate directed hopsets, respectively.

\subsection{Applications}
\label{subsec:app-to-dp}

Our motivation for studying low-sensitivity hopsets is two-fold.
The more concrete motivation is that they are already implicitly used in the problem of differentially private range queries, and our new upper bound for low-sensitivity hopsets (Theorem \ref{thm:greedy-good}) immediately implies improved results for this problem (details below).

More generally, we believe that sensitivity is a natural measure for evaluating hopsets, as it is one way of modeling the ``robustness'' of a hopset: low sensitivity means that changing a vertex/edge in the underlying graph $G$ only affects a small number of shortcuts in $H$.
This corresponds, for example, to the following real-world scenario. In computer networking the notion of an overlay network~\cite{Tarkoma2010-zr} considers logical links layered on top of a physical network.
The logical links support functionality in the overlay protocol but they are actually implemented along a physical path in the underlying network.
Overlay networks have been widely adopted in practice (e.g. VPN, VoIP, content delivery, P2P services) due to benefits such as encapsulation, ease of deployment, and quality of service requirements.
The design of an overlay network could benefit from the resilience of low sensitivity hopsets -- changes in one vertex or one edge in the underlying network only affect relatively few overlay links.

\subsubsection{Differentially Private Range Query}

Low-sensitivity hopsets have (implicitly) found applications in differential privacy (DP)~\cite{dwork2014algorithmic}; this is the problem setting with which we are primarily motivated by.
On a high level, differential privacy protects sensitive information by introducing perturbation, such that the output stays immune to the presence or absence of any individual's data.
More formally,
\begin{definition}[\textit{Differential Privacy}]
    \label{def:dp}
    \normalfont For databases $Y$ and $Y'$ that differ on one data entry and $\epsilon>0, \delta\ge 0$, a randomized algorithm $\mathcal{M}$ is \emph{$(\epsilon,\delta)$-differentially private}, if for any  measurable set $A$ in the range of possible outputs, it holds that:
    $\prob{\mathcal{M}(Y) \in A}\leq e^{\epsilon} \prob{\mathcal{M}(Y') \in A} +\delta.$
\end{definition}
The algorithm $\mathcal{M}$ is called pure-DP if $\delta = 0$ and approximate-DP otherwise\footnote{Note that we are using `epsilon' in two different ways in this paper, but we disambiguate the symbols: $\epsilon$ for the differential privacy parameter, and $\eps$ for the approximation parameter of hopsets.}.
Low-sensitivity hopsets have a direct application to the \emph{All Sets Range Queries (ASRQ)} problem in differential privacy~\cite{deng2023differentially}.
The input to ASRQ is an undirected graph $G$ with public topology and shortest paths, and each edge is associated with a private attribute (which can be distinct from its weight).
The output of ASRQ is an $n\times n$ matrix $M$, where entry $(u,v)$ contains the summation of edge attributes along a shortest path between the vertex pair $u$, $v$.
The DP-ASRQ problem asks for a private mechanism to output a matrix $M'$ with the DP guarantee, while minimizing the utility or additive error, i.e. $\ell_\infty$ of $M-M'$. 
The best-known upper bound of the additive error from prior work is $\Ot(n^{1/3})$ for the pure-DP setting and $\Ot(n^{1/4})$ for the approximate-DP setting.
The latter is tight up to polylogarithmic factors~\cite{bodwin2024discrepancy}.
As we outline next, our low-sensitivity exact hopsets imply an improvement of the $\Ot(n^{1/3})$ bound for the pure-DP setting, down to the tight bound of $\Ot(n^{1/4})$, matching the approximate-DP setting.

We are able to show the connection between low-sensitivity hopsets and the DP-ASRQ problem by the following theorem, which is shown implicitly in~\cite{deng2023differentially}.

\begin{theorem}
    \label{thm:dp-hopset-relation}
    If there exists a $(\beta,0)$-hopset $H$ with edge-sensitivity $\inorm{\esens}$ for any given undirected graph, then for any $\epsilon>0$, there exists an $\epsilon$-DP algorithm for the ASRQ problem such that the additive error is at most $\Ot(\frac{1}{\epsilon}\cdot \sqrt{\inorm{\esens} \cdot \beta})$ with high probability.
\end{theorem}

Plugging into \Cref{thm:dp-hopset-relation} our $(\sqrt{n}\log n, 0)$-hopset with  $\inorm\vsens = O(\log n)$ (and thus $\inorm\esens = O(\log n)$) from \Cref{thm:greedy-good}, we obtain the following result:

\begin{theorem}[$\epsilon$-DP upper bound]
    \label{thm:pure-dp}
    There exists an $\epsilon$-DP algorithm for the ASRQ problem such that the additive error is at most  $\Ot(\frac{n^{1/4}}{\epsilon})$ with high probability. 
\end{theorem}

We show the proof of \Cref{thm:pure-dp} in \Cref{app:dp-rq-sd}.
We remark that low-sensitivity hopsets are also implicitly used in another problem studied in differential privacy: All Pairs Shortest Distances.
However a direct relation like \Cref{thm:dp-hopset-relation} is not available.
We further discuss this problem in \Cref{app:dp-rq-sd}.

\subsection{Organization}

We first provide, in \Cref{sec:tech-overview}, a high-level overview of the technical challenges and ideas in our results.
Formal preliminaries are then given in \Cref{sec:prelim}.
Next, we show an upper bound for undirected exact hopsets and directed shortcut sets in \Cref{sec:upper-greedy} and undirected approximate hopsets in \Cref{sec:upper-apx-undi}.
In \Cref{sec:lb}, we show lower bounds on the sensitivity-diameter tradeoff for shortcut sets and hopsets.
We finish by listing several open problems in \Cref{sec:open}.

\section{Technical Overview}
\label{sec:tech-overview}
In this section we outline the key technical ideas in our constructions.
A natural starting point for proving upper and lower bounds for low-sensitivity hopsets is to simply look at bounds for traditional hopsets.
As it turns out, these ideas can sometimes be massaged into the low-sensitivity setting, but they do not always give the best sensitivity bounds, which necessitates the development of new techniques. 

We begin by discussing a technique from prior work on low-sensitivity hopsets that illuminates a key difference between the traditional and low-sensitivity settings.

\paragraph{Techniques from Prior Work: Heavy-Light Decomposition.}

We will start with a simple example: shortcut sets for undirected graphs.
For traditional shortcut sets, this is a trivial problem.
Simply add a star centered at an arbitrary vertex $v$ to get a shortcut set with $n-1$ edges and diameter only $2$.
If we consider the sensitivity of this construction, we quickly realize that $v$ has sensitivity $n-1$.
Thus, this trivial solution does not work in the low-sensitivity setting.
However, as noted implicitly in prior work \cite{fan2022distances}, there is quite a simple construction that does work. 

The solution is to use a \emph{heavy-light decomposition} of the tree of routing paths rooted at $v$.
A heavy-light decomposition partitions the edges of a tree into ``light'' edges and ``heavy'' paths, where any root-to-leaf path has $O(\log n)$ light edges.
This is useful because we can independently shortcut each of the (vertex-disjoint) heavy paths down to diameter and sensitivity $O(\log n)$ using standard methods (see e.g. \cref{fig:path-shortcut}).
This results in a shortcut set with vertex-sensitivity $O(\log n)$ and diameter $O(\log^2 n)$.
See \Cref{sec:heavylight} for more details. 

This solution suggests a more general technique for constructing low-sensitivity hopsets: whenever a traditional hopset contains a star, simply replace it by the above heavy-light decomposition shortcutting scheme.
This technique works, for instance, to translate the folklore exact hopset into the low-sensitivity setting.
The folklore hopset with $\Ot(n)$ edges and $\Ot(\sqrt{n})$ hop-diameter is the following: randomly sample a set $S$ of $\Ot(\sqrt{n})$ vertices and add a hopset edge between all pairs of reachable vertices in $S$ (where the edge weight is the distance between its endpoints).
Viewing each vertex in $S$ as the center of a star, we can apply the heavy-light decomposition transformation to achieve both vertex-sensitivity and hop-diameter $\tilde{O}(\sqrt{n})$ \cite{ghazi2022differentially,deng2023differentially}.
This remains the best-known result for exact and approximate low-sensitivity directed hopsets.

We remark that it would also be natural to try to adapt Kogan and Parter's recent $\Ot(n^{1/3})$-diameter directed shortcut set~\cite{kogan2022new} on $\Ot(n)$ edges, to the low-sensitivity setting.
In short, it is not clear how to do this.
The first step of their construction picks $\Ot(n^{2/3})$ chains on $\Ot(n^{1/3})$ vertices each, and shortcuts each of them. 
Shortcutting a single chain could add sensitivity to $\Omega(n)$ vertices and edges (since it is a chain, not a path).
So, shortcutting $\Ot(n^{2/3})$ chains could already incur vertex and edge sensitivity $\Ot(n^{2/3})$. 

\paragraph{Our Main Technical Contribution.}
Although the known constructions for low-sensitivity hopsets use the strategy of starting with traditional hopsets and applying the above heavy-light decomposition transformation, it is not clear that this is the optimal strategy.
In particular, the low-sensitivity setting permits us freedom not allowed in the traditional setting: we can add as many edges as we'd like, as long as the sensitivity of each individual vertex is bounded.
Thus, we would like to take advantage of our ability to add many edges, while ensuring that not too many of our added edges have overlapping underlying routing paths.
In light of this new goal, we examine low-sensitivity hopsets from scratch, developing a new technique that is conceptually very different from any known techniques for the traditional setting.

Our main new technique allows us to improve above exact undirected hopset of prior work \cite{ghazi2022differentially,deng2023differentially} from vertex sensitivity $\Ot(\sqrt{n})$ all the way down to $O(\log n)$.
Our resulting hopset with $O(\log n)$ vertex-sensitivity and $\Ot(\sqrt{n})$ hop-diameter is tight up to polylogarithmic factors in both parameters.

This technique heavily relies on the assumption that the routing paths are chosen to be \emph{consistent}; that is, each pair of overlapping paths overlaps at one single contiguous subpath.
This assumption holds for shortest paths in undirected graphs and DAGs (with a consistent tie-breaking scheme), but not for shortest paths in general directed graphs.
This is why our technique does not apply to exact or approximate directed hopsets.
But it does apply to directed shortcut sets when we use the technique on the DAG of SCCs, in combination with other ingredients such as carefully choosing routing paths.

Now, we outline the technique.
It is a greedy algorithm for processing routing paths in a carefully chosen order.
Each time we process a routing path, we shortcut the path using a simple and standard method (see \cref{fig:path-shortcut}) which incurs $O(\log n)$ sensitivity.
We need to be careful because if we shortcut a path that overlaps with a previously shortcutted path, the vertices in the overlap incur double the sensitivity.
For this reason, when we process a path we only shortcut the segments of a path that have not been previously shortcutted.
We think of it this way: every time a path is shortcutted, it cast ``shadows'' on other overlapping paths, and we only shortcut the non-shadowed segments of each path.

Using this method, it is immediate from the standard method for shortcutting a path, that the vertex-sensitivity is only $O(\log n)$.
However, the hop-diameter is in question.
The hop-diameter can be determined by roughly the number of shadows cast onto a path, since each shadow chops the path into more segments, and so the goal is to bound the number of shadows cast onto each path.
Since the routing paths are consistent, we know that when we process a path it only casts at most one new shadowed segment (which may be further segmented by already existing shadows) onto every other path.
This helps, but we still need to choose the order to process the paths carefully.
Even with the consistency property, processing the paths in the wrong order could lead to $\Omega(n)$ vertex-sensitivity. 

To choose the processing order, we carefully define a potential function and choose the path with maximum potential.
The potential of a path $P$ is rather simple to define and fast to implement: the number of vertices on $P$ that are not on any previously chosen path.
In our analysis, we show that this potential function yields an algorithm with the desired hop-diameter of $\tilde{O}(\sqrt{n})$.

\paragraph{Techniques for Lower Bounds.}

For our lower bounds, we use modifications of constructions that have previously been applied to traditional hopsets as well as spanners, emulators, distance preservers, reachability preservers, and related problems.
These graph constructions are layered graphs defined by a collection of \emph{perfect paths}, which are unique paths (or unique shortest paths, depending on the setting) that go from the first to last layer and are pairwise edge-disjoint.
These graphs are useful for traditional hopsets because their properties imply that the addition of a single edge can only decrease the hop-length of one perfect path.
This condition is useful for low-sensitivity hopsets too, but the situation is more nuanced.
When we add a single edge to a traditional hopset it simply counts $1$ towards our budget of edges, whereas in the low-sensitivity setting the total amount of sensitivity added depends on the ``length'' of the added edge.
Thus, we need to take into account all possible edge lengths, and consider their contributions to both sensitivity and diameter. 

As a separate issue, we are interested in edge sensitivity for lower bounds.
For this reason we modify known perfect path constructions by replacing each vertex by an edge, which requires a slightly more careful analysis. 

Our final construction begins with a known perfect path construction and optimizes the parameters for the low-sensitivity hopset setting (which results in different graph parameters than constructions for traditional hopsets).
We obtain a lower bound showing that to achieve polylogarithmic vertex or edge-sensitivity for exact undirected hopsets, the hop-diameter must be $\widetilde{\Omega}(\sqrt{n})$, which is tight with our aforementioned upper bound.
We also obtain lower bounds for low-sensitivity directed shortcut sets and low-sensitivity approximate undirected hopsets, both by again appealing to known constructions of layered graphs.

\begin{remark}
\label{rem:optimal-hopset-lb}
    Bodwin and Hoppenworth \cite{bodwin2023folklore} are able to prove stronger lower bounds on the number of edges in a hopset by relaxing the requirement that perfect paths are disjoint.
    Unfortunately, their relaxation does not immediately seem to be useful for lower bounding the \textit{sensitivity} of a hopset; see \Cref{app:optimal-hopset-lb} for more details. 
\end{remark}

\paragraph{Techniques for Approximate Undirected Hopsets. }

For traditional hopsets with $O(n)$ size, there are essentially 3 diameter regimes: 
(1) directed shortcut/hopsets as well as undirected exact hopsets all require polynomial hop-diameter, 
(2) \emph{approximate} hopsets for undirected graphs require only $n^{o(1)}$ hop-diameter, and 
(3) undirected shortcut sets trivially achieve diameter $2$.
In the low-sensitivity setting, we have mainly discussed the polynomial diameter regime so far, but it also makes sense to consider the $n^{o(1)}$-hop-diameter regime for low-sensitivity approximate undirected hopsets.
To this end, we extend known results for traditional approximate undirected hopsets to the low-sensitivity setting, achieving both vertex-sensitivity and hop-diameter $n^{o(1)}$ for any $\eps > 0$ that is at least inverse polylogarithmic.

In the traditional setting, Huang-Pettie and Elkin-Neiman \cite{huang2019thorup,elkin2019hopsets} showed that Thorup-Zwick emulators \cite{thorup2006spanners} are optimal hopsets.
Examining this construction quickly reveals that it has high vertex and edge sensitivity.
Instead, we begin with the similar construction of Thorup-Zwick spanners (also in \cite{thorup2006spanners}).
Recall that emulators and spanners are both sparse graphs that approximately preserve distances in graphs, but a spanner is a subgraph while an emulator can have added edges.
Thus, it makes sense to consider an emulator as a hopset, but not a spanner (since it has no added edges).
However, we can use the heavy-light tree decomposition shortcutting on top of the Thorup-Zwick spanner to obtain a construction that is similar to the Thorup-Zwick emulator, but now it has low sensitivity.
Then, we can leverage the hopset analysis tools of Huang-Pettie to show that our $n^{o(1)}$ vertex sensitivity hopset also has hop-diameter $n^{o(1)}$ for any $\eps > 0$ that is at least inverse polylogarithmic.

\section{Preliminaries}
\label{sec:prelim}

\subsection{Notations and Definitions}
\label{subsec:formal-def-hopset}
With the aim of formally defining vertex and edge sensitivity, we first introduce a couple of notions.

\begin{definition}[\textit{Routing Paths}]
\label{def:routing-paths}
    \normalfont Given a graph $G = (V, E)$ we call $\cP$ a set of \emph{routing paths} for $G$ if for all $s,t \in V$ such that $s$ can reach $t$, there is exactly one path $P \in \cP$ such that $P$ is a path in $G$ with starting point $s$ and ending point $t$.
    If $G$ is weighted, we stipulate that the paths in $\cP$ must be shortest paths in $G$.
\end{definition}

Routing paths are meant to model paths that are deemed critical.
In the context of hopsets, the routing paths must thus be shortest paths, whereas in the context of shortcut sets, any path may be chosen as a routing path.

\begin{definition}[\textit{Span}]
\label{def:span}
    \normalfont Given a graph $G = (V, E)$ and a tuple of a set of routing paths and a shortcut/hopset $\cA(G) = (\cP, H)$, we define the \emph{span} of an edge $e = (s,t) \in H$ to be the unique path in $\cP$ with starting point $s$ and ending point $t$; we denote this with $\edgespan_{\cA(G)}(e)$.
    When the context is clear, we omit the subscript and write $\edgespan(e)$.
\end{definition}

We are now ready to define vertex and edge sensitivity.
These are properties of (the inputs) graphs, together with (the outputs) a set of routing paths and a shortcut/hopset.

\begin{definition}[\textit{Vertex Sensitivity}]
\label{def:vertex-sensitivity}
    \normalfont Given a graph $G = (V, E)$ and a tuple of a set of routing paths and a shortcut/hopset $\cA(G) = (\cP, H)$, the \emph{vertex sensitivity} vector, indexed by $V$, is denoted with $\vsens_{G,\cA(G)}$.
    When the context is clear, we omit any subset of the subscripts: $\vsens$.
    
    The value of $\vsens$ at index $v$ is
    \[\vsens(v) = \card{\set{f \in H : v \in V(\edgespan(f))}}.\]
    The worst case \emph{vertex sensitivity} (given $G, \cA(G)$) is
    \[\inorm{\vsens} = \max_{v \in V} \vsens(v).\]
\end{definition}

\begin{definition}[\textit{Edge Sensitivity}]
\label{def:edge-sensitivity}
    \normalfont Given a graph $G = (V, E)$ and a tuple of a set of routing paths and a shortcut/hopset $\cA(G) = (\cP, H)$, the \emph{edge sensitivity} vector, indexed by $E$, is denoted with $\esens_{G,\cA(G)}$.
    When the context is clear, we omit any subset of the subscripts: $\esens$.
    
    The value of $\esens$ at index $e$ is
    \[\esens(e) = \card{\set{f \in H : e \in E(\edgespan(f))}}.\]
    The worst case \emph{edge sensitivity} (given $G, \cA(G)$) is
    \[\inorm{\esens} = \max_{e \in E} \esens(e).\]
\end{definition}

A reason behind our choice of using vector notation to describe sensitivity is as follows: We sometimes consider the total sensitivity $\onenorm{\esens} = \sum_e \esens(e)$ and even the sum of sensitivities over a set $T \subseteq E$ which we write with $\mathds{1}_{T} \cdot \esens$ where $\mathds{1}_{T}$ is the indicator vector of $T$ and `$\cdot$' denotes the usual inner product over vectors.

Next, we define the notion of \emph{consistency}, which will prove to be a very useful property of carefully chosen routing paths in our upper bound constructions.

\begin{definition}[\textit{Consistency} (Modification from \cite{coppersmith2006sparse})]
\label{def:consistent}
    \normalfont A pair of paths $P, P'$ are said to be \emph{consistent} if their intersection contains at most one connected component.
    A set of paths $\cP$ is said to be \emph{consistent} if every pair of paths $P, P' \in \cP$ are consistent.
    
    A shortest path tiebreaking function chooses for each pair of vertices $s,t \in V$, where $s$ can reach $t$, exactly one shortest path from $s$ to $t$.
    We say a shortest path tiebreaking function is \emph{consistent} if its image is consistent.
\end{definition}

Our lower bounds go through via constructions of graphs inspired by \cite{bodwin2021new, abboud2018hierarchy}; in fact we directly use Theorems~5 and 6 from the first paper and Theorem~4.6 from the second paper for our constructions, but our proofs and the way we invoke the theorems differ.
A common object that appears in both papers, and that we use, are layered graphs.

\begin{definition}[\textit{Layered Graphs}]
\label{def:layered-graph}
    \normalfont A graph $G=(V,E)$ is \emph{layered} if the vertex set can be partitioned into $V = V_1 \sqcup V_2 \sqcup \ldots \sqcup V_\ell$ such that every edge goes between two consecutive parts $V_i$ and $V_{i+1}$.
    If $G$ is directed, then edges are oriented towards the parts with larger indices.
\end{definition}

\subsection{Relationship Between Vertex and Edge Sensitivity}
\label{subsec:sens-transfer}
We discuss how different notions of sensitivity relate to each other.
First, observe that for a fixed shorctut/hopset $H$, its $\inorm{\vsens}$ and $\inorm{\esens}$ could differ drastically.
For example, consider an $n/2$-tipped star with tips of length $2$; any $H$ comprising of all edges from the center $c$ to the ends of all the tips would have $\vsens(c) = n/2$ but $\esens(e) = 1$ for all edges $e$.
More generally, the shortcut/hopset edges with endpoint $v$ contribute one each to $\vsens(v)$ whereas their contribution to the edge sensitivity can be dispersed over the routing paths incident to $v$.
At any rate, observe that the edge sensitivity being no more than the vertex sensitivity generalizes to any graph.

\begin{observation}
\label{obs:esens-less-vsens}
    $\inorm{\esens_{G,\cA(G)}} \le \inorm{\vsens_{G,\cA(G)}}$.
\end{observation}
\begin{proof}
    Consider the edge $e = (u,v)$ where $\esens(e) = \inorm{\esens}$.
    Observe in particular that $\esens(e) \le \vsens(u)$ since every edge in $\cA(G)$ contributing to $\esens(e)$ also contributes to $\vsens(u)$.
    We therefore conclude that $\inorm{\esens_{G,\cA(G)}} \le \inorm{\vsens_{G,\cA(G)}}$.
\end{proof}

It follows from \Cref{obs:esens-less-vsens} that proving upper bounds for $\inorm{\vsens}$ gets us the same bound for $\inorm{\esens}$ and, in the contrapositive, proving lower bounds for $\inorm{\esens}$ gets us the same bound for $\inorm{\vsens}$.
We thus strive to prove upper bounds on $\inorm{\vsens}$ and lower bounds on $\inorm{\esens}$ throughout the paper; all results presented here are of this form.

\subsection{Algorithmic Building Blocks}
In this section we describe primitive subroutines used in our constructions of shortcut/hopsets.
The first tool addresses how to ``shorten'' paths; the second, which can be viewed as a generalization of the first, addresses how to shorten trees.

\subsubsection{Path Shortcutting}

\pathshortcut is a well-known primitive which, when given a path, reduces hop-lengths of shortest paths to $O(\log n)$ while adding no more than $O(\log n)$ sensitivity to any vertex.
Intuitively, \pathshortcut adds an edge between the two endpoints of the path, then partitions the path into two equally sized subpaths and recurses into each of them.
An example of the output of \pathshortcut is depicted in \Cref{fig:path-shortcut}.

\begin{tbox}
    \algname{\pathshortcut} \\
    \textbf{Input:} A path $P$ from $v_0$ to $v_\ell$
    \begin{enumerate}
        \item If $\ell < 2$ then return $\{\}$
        \item Return $\{(v_0, v_\ell)\} \cup \pathshortcut(P[v_0..v_{\lfloor \ell/2\rfloor}]) \cup \pathshortcut (P[v_{\lfloor \ell/2\rfloor}..v_\ell])$
    \end{enumerate}
\end{tbox}

\begin{figure}[h!]
    \centering
    \includegraphics[scale=0.4]{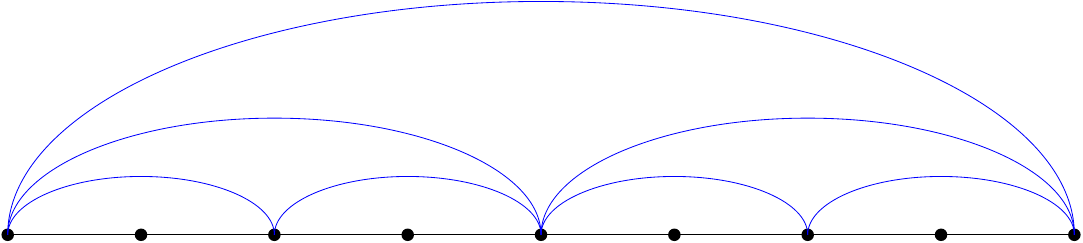}
    \caption{An example of \pathshortcut on a path of length 8 with vertices $v_0, \dots, v_8$. The collection of blue edges is the shortcut/hopset output.}
    \label{fig:path-shortcut}
\end{figure}

\begin{restatable}{observation}{obspathshortcut}
\label{obs:path-shortcut}
    Given a path $P = v_0, v_1, \ldots, v_n$ (possibly oriented in the $(v_i, v_{i+1})$ direction), \pathshortcut outputs a $O(\log n)$-shortcut/hopset with $\inorm{\vsens} = O(\log n)$.
\end{restatable}

For completeness, we provide a proof in \Cref{app:path-tree}.

\subsubsection{Tree Shortcutting via the Heavy-light Decomposition}
\label{sec:heavylight}

\treeshortcut, described in \cite{fan2022distances}, reduces the problem of shortcutting trees to that of shortcutting paths, where the instances passed to \pathshortcut are found via a heavy-light decomposition of the tree.
A heavy-light decomposition~\cite{harel1984fast} (also termed heavy path decomposition) of a rooted tree is a partitioning of its edges into `heavy' paths and `light' edges such that any root-to-leaf path passes through at most $O(\log n)$ light edges, which is captured by \Cref{prop:hl-decom}.

\begin{proposition}[\cite{harel1984fast}]
    \label{prop:hl-decom}
    Given a tree $T$ of size $n$, there is a heavy-light decomposition such that
    \begin{itemize}
        \item any root-to-leaf path has at most $O(\log n)$ light edges;
        \item all heavy paths are vertex disjoint.
    \end{itemize}
\end{proposition}

Using the above proposition, \treeshortcut is specified as follows.

\begin{tbox}
    \algname{\treeshortcut} \\
    \textbf{Input:} A tree $T$ rooted at $v$
    \begin{enumerate}
        \item Return $\pathshortcut(P)$ for all heavy paths $P$ in a heavy-light decomposition of $T$
    \end{enumerate}
\end{tbox}

\begin{restatable}{observation}{obstreeshortcut}
\label{obs:tree-shortcut}
    Given a tree $T$ rooted at $v$ (with possibly all edges oriented away from the root or all edges oriented towards the root), \treeshortcut outputs an $O(\log^2 n)$-shortcut/hopset with $\inorm{\vsens} = O(\log n)$.
\end{restatable}

This result is implicitly shown in the analysis of~\cite{fan2022distances}, but we provide a proof (in a differential-privacy-free language) in \Cref{app:path-tree} for completeness.
Elementary use of \treeshortcut yields the following quick results.
A $O(\log^2 n)$-shortcut set construction with $\inorm{\vsens} = O(\log n)$ on undirected graphs.
A directed $(O(\sqrt{n} \log n),0)$-hopset construction with $\inorm{\vsens} = O(\sqrt{n} \log n)$.
See \Cref{app:warmup-1,app:warmup-2} for details.

\section{Upper Bounds via a New Greedy Approach}
\label{sec:upper-greedy}
This section describes a new way to construct low-sensitivity shortcut/hopsets when we are able to choose consistent routing paths; this is always possible, for example, in undirected graphs and DAGs.
Using this approach, we can quite immediately get \Cref{thm:greedy-good}.
With some slight care, we apply this approach to shortcut sets on directed graphs and obtain \Cref{thm:direted-reach-ub}.

\subsection{The Greedy Construction}

At a schematic-level, the construction is very simple to describe: for each routing path $P$, processed in \emph{some order}, apply \pathshortcut to the \emph{uncovered pieces} of $P$.
The details are filled in as follows.

\paragraph{Finding an Order to Process Paths. }
The order for which we process the shortest paths is described via a potential function $\Phi$ on routing paths $P$, which counts the number of vertices in $P$ that have not yet been contained in a path that has thus far been processed.
\begin{definition}
\label{def:pot-1}
    \normalfont Let $\cP$ be the set of all routing paths, and let $\cP' \subseteq \cP$ be a subset of the set of all routing paths.
    We define, for all $P \in \cP$, the potential
    \[\Phi_{\cP'}(P) = \card{\set{v \in V(P) : v \not\in V(P') \textrm{ for all paths } P' \in \cP'}}.\]
    If it is clear from the context, we omit the subscript and write $\Phi(P)$.
\end{definition}
The $i^{th}$ routing path in the order is then the one which maximizes $\Phi_{\cP^{(i-1)}}$, where $\cP^{(i-1)}$ is the set of the first $i-1$ routing paths in the order (that is, the ones that have already been processed).

The following is an observation that we will make use of later; it says that the potential of a path $P$ is weakly decreasing throughout the construction process, which can be seen by noting that if $P$ is processed then it has $0$ potential and, otherwise, its vertices that are not contained in already processed paths can only drop as more and more paths are processed.
\begin{observation}
\label{obs:pot-monotone}
    $\Phi_{\cP^{(i)}}$ is weakly decreasing in $i$.
\end{observation}

\paragraph{Processing a Path. }
The uncovered pieces of an unprocessed path $P$ depend on the the history of the algorithm up to the point where $P$ is processed.
These are the maximal subpaths of $P$ for which no vertex is part of an already processed path.
\begin{definition} [\textit{Uncovered pieces of a shortest path}]
\label{def:relevant-pieces}
    \normalfont Let $\cP'$ be the set of paths processed over so far.
    The \emph{uncovered pieces} of $P$, given $\cP'$, are then the maximal connected components of $\set{v \in V(P) : v \not\in V(P') \textrm{ for all paths } P' \in \cP'}$.
\end{definition}
A path $P$ is processed by running \pathshortcut on each of its uncovered pieces.

\paragraph{Putting Things Together. }
The algorithm producing an $\Ot(\sqrt{n},0)$-hopset with $\inorm{\vsens} = O(\log n)$ is then described by \greedyshortcut which selects the path that maximizes the potential $\Phi$ and processes it as above, repeating until every path is processed.
\begin{tbox}
    \algname{\greedyshortcut}\\
    \textbf{Input:} A graph $G = (V, E, w)$ with a consistent set of routing paths $\cP$.
    \begin{enumerate}
        \item $\cP' \gets \emptyset, H \gets \emptyset$
        \item While $\card{\cP'} < {n \choose 2}$
            \begin{enumerate}
                \item Select the next routing path $P^* \gets \argmax\limits_{P \in \cP \setminus \cP'} \Phi_{\cP'}(P)$, breaking ties arbitrarily
                \item $H \gets H \cup \pathshortcut(p)$ for all uncovered pieces $p$ of $P^*$
                \item $\cP' \gets \cP' \cup \set{P^*}$
            \end{enumerate}
        \item Return $H$
    \end{enumerate}
\end{tbox}

One can quickly observe that $\greedyshortcut$ produces a low-sensitivity shortcut/hopset.

\begin{observation}
\label{obs:greedy-sens}
    The hopset produced by \greedyshortcut has $\inorm{\vsens} = O(\log n)$.
\end{observation}
\begin{proof}
    Let $v$ be an arbitrary vertex.
    Its sensitivity is only ever increased in one iteration of \greedyshortcut, namely the first time a path $P$ is chosen such that $v \in V(P)$.
    We know from \Cref{obs:path-shortcut} that the contribution to $\vsens(v)$ in one iteration is bounded by $O(\log n)$, completing the proof.
\end{proof}

Before we turn to the remaining claim that \greedyshortcut is an $\Ot(\sqrt{n}, 0)$-hopset, we need a couple of definitions.
Henceforth, let the paths be processed in the order $P_1, P_2, \ldots, P_{O(n^2)}$.
First we define the \emph{shadows} and \emph{penumbras} cast onto $P_i$.
Shadows are, loosely speaking, complementary to the uncovered pieces of $P_i$ when it is processed; they serve two main purposes: (i) to bound the number of uncovered pieces when $P_i$ is processed, and (ii) as shortcuts that can be used to traverse the covered pieces of $P_i$.

\begin{definition}[\textit{Shadows, Penumbras}]
   \normalfont For any $j < i$, the \emph{shadows} cast by $P_j$ onto $P_i$ are the segments of $P_j$ that coincide with $P_i$ after removing all shadows cast before $P_j$ onto $P_i$.
    Formally, they are the maximal subpaths of
    \begin{align*}
        P_i^{(j)} = \set{v \in P_i \mid v \in P_j \textrm{ and } v \not\in P_k \textrm{ for } k < j},
    \end{align*}
    defined for all $j < i$.

    Let $S$ be a shadow cast onto $P_i$.
    We call the edges in $E(S, P_i \setminus S)$ the \emph{penumbras} of $S$ cast onto $P_i$, where $E(S, P_i \setminus S)$ denotes the edges with one endpoint in $S$ and the other in $P_i \setminus S$.
\end{definition}

Observe that a single path $P_j$ may cast $\Omega(j)$ shadows onto $P_i$.
Nevertheless, if $P_i$ has many shadows cast onto it, there must be almost as many paths which cast shadows onto $P_i$.

\begin{observation}
\label{obs:lots-of-shadows-lots-of-casters}
    Suppose after iterating $P_1$ through to $P_t$, there are $k$ shadows cast onto $P_i$ for $i > t$.
    Then, at least $(k-1)/2$ distinct paths among $P_1$ through to $P_t$ cast at least one shadow onto $P_i$.
\end{observation}
\begin{proof}
    First, note that since there are $k$ shadows cast onto $P_i$, there are at least $k-1$ edges in $P_i$ which are penumbras.
    While the path $P_j$ may cast many penumbras onto $P_i$ for $j < i$, by the consistency of routing paths $P_j$ casts at most $2$ penumbra onto $P_i$ that have not already been cast before $P_j$; these edges are namely the ones incident to boundary of the intersection between $P_j$ and $P_i$ (see \Cref{fig:penumbra-charging}).
    We conclude that there are at least $(k-1)/2$ distinct paths among $P_1$ through to $P_t$ that cast at least one shadow onto $P_i$.
\end{proof}

\begin{figure}[h]
    \centering
    \includegraphics[scale=0.7]{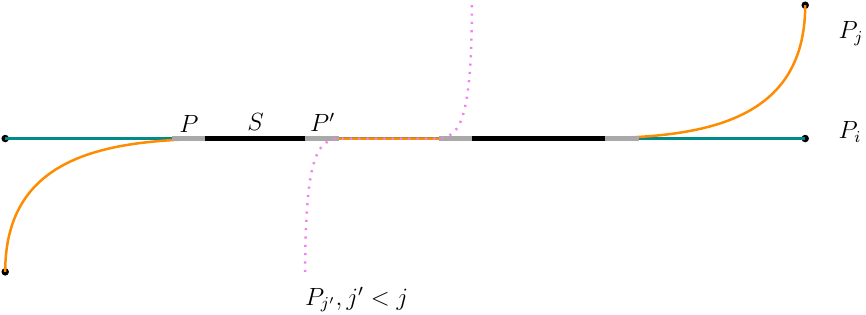}
    \caption{The penumbra $P$ of $S$ (a shadow cast by $P_j$ onto $P_i$) is new, but the penumbra $P'$ of $S$ is also the penumbra of a shadow cast by $P_{j'}$ and is thus not new.
    }
    \label{fig:penumbra-charging}
\end{figure}

Suppose a routing path $P$ from $s$ to $t$ has no more than $k$ shadows cast onto it.
Then the $G \cup \greedyshortcut(G)$ hop-length from $s$ to $t$ is at most $\Ot(k)$; each shadow has shortcut-distance at most $O(\log n)$ by an application of \Cref{obs:path-shortcut} to the time the shadow was first cast, and the remaining parts of $P$, of which there are $O(k)$ of them, each have shortcut-distance at most $O(\log n)$ by another application of \Cref{obs:path-shortcut} at the time $P$ is processed.
Motivated by this, we now show that every routing path has at most $O(\sqrt{n})$ shadows cast onto it and every shadow on $P$ has hop-length $O(\log n)$.

\begin{lemma}
\label{lem:bounded-shadows}
    $P_i$ has $O(\sqrt{n})$ shadows cast onto it, for all $i$.
\end{lemma}
\begin{proof}
    Suppose, for a contradiction, that there is a routing path $P_i$ which has at least $4 \sqrt{n}$ shadows cast onto it.
    We aim to show that there are many ``highly disjoint'' paths that cast at least one shadow onto $P_i$, which will lead to a contradiction since too many such paths imply $G$ has more than $n$ vertices.

    Towards this, let $P_t$ for $t < i$ be the first path processed after which $\Phi(P_i) < \sqrt{n}$.
    Such a $t$ must exist since, otherwise, we can invoke the monotonicity of $\Phi$, which is weakly decreasing (see \Cref{obs:pot-monotone}), and \Cref{obs:lots-of-shadows-lots-of-casters}, which says that there must be at least $(4 \sqrt{n} - 1)/2$ distinct paths casting shadows onto $P_i$, to show that $G$ has $\sqrt{n} \cdot (4 \sqrt{n} - 1)/2 > n$ vertices which is a contradiction.

    After $P_t$ is processed, there can be at most $\sqrt{n}$ more shadows cast onto $P_i$ since $\Phi(P_i) < \sqrt{n}$ and each shadow cast onto $P_i$ reduces $\Phi(P_i)$ by at least $1$.
    There are therefore at least $3 \sqrt{n}$ shadows cast onto $P_i$ by paths processed before $P_t$ (inclusive).
    Then, we can repeat the same argument as before: by \Cref{obs:pot-monotone} and \Cref{obs:lots-of-shadows-lots-of-casters} $G$ has $\sqrt{n} \cdot (3 \sqrt{n} - 1)/2 > n$ vertices, leading to a contradiction.
\end{proof}

\begin{observation}
\label{obs:short-shadows}
    Let $P$ be a routing path, and let $S$ be any shadow with endpoints $s,t$ cast onto $P$.
    Then the $G \cup \greedyshortcut(G)$ hop-length from $s$ to $t$ (using edges contained in $S$ or whose span is contained in $S$)\footnote{This is an important constraint, since we wish for this claim to apply to $(\beta, 0)$-hopsets} is bounded by $O(\log n)$.
\end{observation}
\begin{proof}
    Let $Q$ be the routing path casting $S$ onto $P$.
    The claim goes through since at the time $Q$ is processed, $S$ is contained in an uncovered piece of $Q$; \greedyshortcut makes a call to \pathshortcut on the uncovered piece of $Q$ containing $S$ when $Q$ is being processed, finishing things up by \Cref{obs:path-shortcut}.
\end{proof}

\subsection{An Upper Bound for Undirected Exact Hopsets}
\label{subsec:undi-hop-upper-proof}
We are able to get upper bounds for exact hopsets on undirected graphs and DAGs immediately using \greedyshortcut.

\greedygood*
\begin{proof}
    We show the theorem for \greedyshortcut.
    This follows from combining \Cref{lem:bounded-shadows} and \Cref{obs:path-shortcut}, and also \Cref{obs:short-shadows}.
    To spell things out, let $G$ be an undirected or directed acyclic graph.
    Any routing path in $G$ has $O(\sqrt{n})$ shadows, each shadow can be crossed within $O(\log n)$ hops (in the correct direction since $G$ is undirected or a DAG), and the uncovered pieces of $P$ at the time it is being processed (of which there are $O(\sqrt{n})$ many) can be crossed within $O(\log n)$ hops since \greedyshortcut makes a call to \pathshortcut on each of these pieces.

    The sensitivity claim follows directly from \Cref{obs:greedy-sens}.
\end{proof}

\subsection{An Upper Bound for Directed Shortcut Sets}
\label{subsec:di-short-upper-proof}
Recall that shortcut sets pertain to the problem of reachability.
The matter of finding low-sensitivity directed shortcut sets can be resolved by a more careful application of \greedyshortcut.
The crux here is that now we have to be slightly careful about our choice of routing paths in order to apply \greedyshortcut effectively; directed graphs may unavoidably contain inconsistent paths (see \Cref{fig:directed-never-consistent}).

\begin{figure}[h]
    \centering
    \includegraphics[scale=0.7]{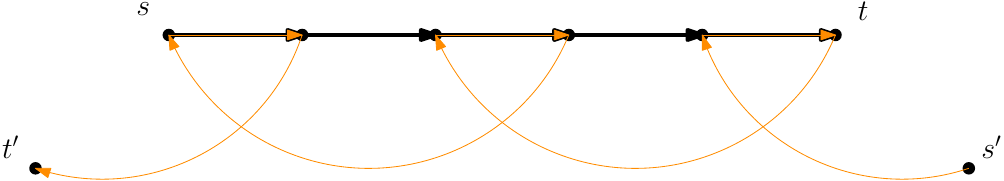}
    \caption{No choice of an $st$-path and an $s't'$-path yields a consistent set of routing paths.}
    \label{fig:directed-never-consistent}
\end{figure}

\paragraph{Routing Paths. }
Fix a directed graph $G = (V, E)$ and consider its strongly connected components (SCCs) $C^1, C^2, \ldots, C^k$.
Our shortcut set construction will contain two types of edges: (i) intra-SCC edges, and (ii) inter-SCC edges.

To handle the intra-SCC edges of $C^i$, we arbitrarily pick exactly one representative vertex $v^i \in C^i$ and use any shortest path tiebreaking function to find a shortest path arborescence of $G[C^i]$, the graph induced by $C^i$, towards $v^i$ and another away from $v^i$.
We denote these arborescences by $T^{in}_{v^i}$ and $T^{out}_{v^i}$ respectively.
All paths in $T^{in}_{v^i}$ from descendent towards ancestor and in $T^{out}_{v^i}$ from ancestor towards descendent are included in the set of routing paths.
Notice that the routing paths are well defined.
For a pair of vertices $s,t \in C^i$ it might seem that we have mistakenly defined two distinct routing paths (one from $T^{in}_{v^i}$ and the other from $T^{out}_{v^i}$); this is, however, not possible since we have used only one shortest path tiebreaking function from which there is a unique shortest $st$-path.

To handle the inter-SCC edges, for any $C^i$ and $C^j$ with $E(C^i, C^j) \neq \emptyset$, we arbitrarily pick exactly one representative edge $e^{i,j} = (u,v) \in E(C^i, C^j)$ and call $u$ an out-port and $v$ an in-port, which we access with $\outport(e^{i,j})$ and $\inport(e^{i,j})$ respectively.
The shortcut set will only ever add inter-SCC edges from out-ports to in-ports and so only routing paths for these types of edges are left to be defined (the other routing paths can be defined arbitrarily).
Let $D$ be the directed acyclic graph (DAG) formed by contracting the SCCs of G, and removing for all $i,j$ parallel edges between the super-vertices $C^i$ and $C^j$ except for the representative edge $e^{i,j}$, if it exists. 
Let $\cP_D$ be the set of shortest paths of $D$ over a \emph{consistent} shortest path tiebreaking function (see \Cref{def:consistent}).
Then for any $P \in \cP_D$, where $P = e_1, e_2, \ldots, e_\ell$, we define the routing path in $G$ from $\outport(e_1)$ to $\inport(e_\ell)$ by connecting, for all $i \in [\ell-1]$, $\outport(e_i)$ to $\inport(e_{i+1})$ via an arbitrary simple path contained in the SCC they both belong to (see \Cref{fig:directed-reach-path}).

\begin{figure}[h]
    \centering
    \includegraphics[scale=1.3]{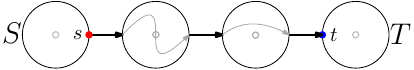}
    \caption{The $ST$-path in $D$ (super-vertices, which are SCCs, are represented by large circles) yields an $st$-path in $G$ where $s$ is an out-port and $t$ is an in-port. Gray paths contained in SCCs are chosen arbitrarily.}
    \label{fig:directed-reach-path}
\end{figure}

\paragraph{The Construction. }
Having defined the relevant routing paths to our shortcut set, the complete construction is as follows.
\begin{tbox}
    \algname{\directedreach}\\
    \textbf{Input:} A directed graph $G = (V, E)$ with a representative vertex $v^i$ for each strongly connected component $C^i$ and a representative edge for each pair of adjacent SCCs.
    \begin{enumerate}
        \item $H \gets \emptyset$
        \item $H \gets H \cup \treeshortcut(T^{in}_{v^i}) \cup \treeshortcut(T^{out}_{v^i})$ for every strongly connected component $C^i$ where $T^*_{v^i}$ is a shortest path arborescence of $G[C^i]$ rooted at $v^i$
        \item $D \gets $ contraction of SCCs of $G$, only keeping representative edges
        \item For every $e \in \greedyshortcut(D)$ update $H \gets H \cup \set{(\outport(e_1), \inport(e_\ell))}$ where $e_1, e_\ell$ are defined by $\edgespan_D(e) = e_1, e_2, \ldots, e_\ell$
        \item Return $H$
    \end{enumerate}
\end{tbox}

It remains to bound the diameter and sensitivity of \directedreach.

\directedreachub*
\begin{proof}\ \\
    We show the theorem for \directedreach.\\
    \textbf{Output is a $O(\sqrt{n}\log^3 n)$-shortcut set. }
    Let $G$ be a directed graph and suppose the vertex $s$ reaches the vertex $t$ in $G$.
    Denote the SCC of $s$ by $S$ and that of $t$ by $T$.
    There must be a shortest path from $S$ to $T$ in the directed acyclic graph $D$.
    Thus, $D \cup \greedyshortcut(D)$ produces an $O(\sqrt{n}\log n)$-shortcut set for $D$, by \Cref{thm:greedy-good}.
    This translates to a path in $G \cup \directedreach(G)$ that uses $O(\sqrt{n}\log n)$ edges to connect $S$ to $T$ (not counting the edges required the move from in-port to out-port within each SCC).
    Moving from an in-port to an out-port of an SCC $C$ takes $O(\log^2 n)$ by the call to $\treeshortcut(C)$, and using \Cref{obs:tree-shortcut}.
    The $G \cup \directedreach(G)$ hop-length from $s$ to $t$ is thus $O(\sqrt{n}\log n) \cdot O(\log^2 n) = O(\sqrt{n}\log^3 n)$.
    \\
    \textbf{Sensitivity is bounded by $O(\log n)$. }
    Consider any vertex $v$.
    The contributions to $\vsens(v)$ come from the call to $\treeshortcut(C)$ where $v \in C$ (the intra-SCC edges) and from the shortcut set edges added by the call to $\greedyshortcut(D)$ (the inter-SCC edges).
    The contribution of $\treeshortcut(C)$ is $O(\log n)$ by \Cref{obs:tree-shortcut}.
    The contribution by the latter is also $O(\log n)$ since any inter-SCC edge which contributes to $\vsens(v)$ has a distinct corresponding edge which contributes to $\vsens_D(C)$ in the call to $\greedyshortcut(D)$, and we know that $\inorm{\vsens_D} = O(\log n)$ by \Cref{obs:greedy-sens}.
\end{proof}
\section{Approximate Undirected Low-Sensitivity Hopset Constructions}
\label{sec:upper-apx-undi}
Here we give an $(n^{o(1)}, \eps)$-hopset with vertex sensitivity $n^{o(1)}$ for any $\eps > 0$ that is at least inverse polylogarithmic.
As a reminder, the main theorem of this section is as follows.
\apxhopsetub*

\paragraph{Intuition. }
Huang-Pettie and Elkin-Neiman showed in \cite{huang2019thorup,elkin2019hopsets} that the emulators constructed by Thorup and Zwick in \cite{thorup2006spanners} are $\paren{O\paren{(k/\eps)^k}, \eps}$-hopsets for every $\eps > 0$.
Briefly, the emulators are constructed by placing vertices in a hierarchy of roughly $k$ levels and, for each vertex $v$, connecting $v$ with a direct edge to all vertices in $\ball[v]$, where $\ball[v]$ are the vertices in $v$'s level that are closer to $v$ than $v$ is to the level one higher than $v$'s (formal details to come later).

Ideally, setting $k = \Theta(\sqrt{\log n})$, the emulator would immediately give us an $(n^{o(1)}, \eps)$-hopset.
But things are hardly ever ideal; for one, vertices in the last level of the hierarchy are connected directly to all of $V$ and thus would witness high sensitivity.
To circumvent this, we construct a low-sensitivity hopset heavily based on (if not entirely lifted from) \cite{thorup2005approximate}, another result of Thorup and Zwick pertaining to distance oracles (and, later, also to spanners in \cite{thorup2006spanners}).
We then use the emulator of Thorup and Zwick as an analysis tool: for any pair $(u,v)$ in the emulator, we show that there is a shortest path using $O(\log^2 n)$ hops to get from $u$ to $v$ in our hopset.

\paragraph{The Thorup-Zwick Spanner Based Construction. }
For the remainder of this section, we work with an undirected weighted graph $G = (V, E, w)$.
We also set the routing paths under a consistent shortest path tiebreaking function (see \Cref{def:consistent}) and refer to them as \emph{the} shortest paths.
To proceed, we first make a definition (to the extent possible, all definitions in this section coincide with those in \cite{thorup2005approximate,thorup2006spanners,huang2019thorup}).

\begin{definition}[\textit{Cluster}]
\label{def:cluster}
    Let $V = A_0 \supseteq A_1 \supseteq \ldots \supseteq A_{k-1}$ be a hierarchy of vertices, and $A_k = \emptyset$.
    The \emph{cluster} of $u \in A_{i} \setminus A_{i+1}$, defined for all $i \in [0,k-1]$, is
    \[\cluster(u) = \set{v \in V : \dist_w(v,u) < \dist_w(v,A_{i+1})}.\]
\end{definition}

Note that if $u \in A_{k-1}$, then $\cluster(u) = V$.
We first give an observation about clusters in advance, to be used later in the analysis of our construction.

\begin{observation}
\label{obs:clus-trunc}
    For all $u \in V$, there is some truncation $T^*_u$ of the shortest path tree $T_u$ rooted at $u$ such that $V(T^*_u) = \cluster(u) \cup \set{u}$.
\end{observation}
\begin{proof}
 Let $u \in A_{i} \setminus A_{i+1}$ and $v \in \cluster(u)$.
    We establish this claim by showing that all vertices along the shortest path from $u$ to $v$ also belong to $\cluster(u)$.
    Let $v'$ be any such vertex, and suppose for a contradiction that $v' \not\in \cluster(u)$.
    Then
    \begin{align*}
        &
        \dist_w(v', A_{i+1}) \le \dist_w(v',u)
        \tag{$v' \not\in \cluster(u)$}
        \\
        \implies
        &
        \dist_w(v, v') + \dist_w(v', A_{i+1}) \le \dist_w(v, v') + \dist_w(v',u)
        \\
        \implies
        &
        \dist_w(v, v') + \dist_w(v', A_{i+1}) \le \dist_w(v, u)
        \tag{$\dist_w(v, v') + \dist_w(v',u) = \dist_w(v, u)$}
        \\
        \implies
        &
        \dist_w(v, A_{i+1}) \le \dist_w(v, u)
        \tag{Triangle inequality: $\dist_w(v, A_{i+1}) \le \dist_w(v, v') + \dist_w(v', A_{i+1})$}
        \\
        \implies
        &
        v \not\in \cluster(u),
    \end{align*}
    which is a contradiction.
\end{proof}

We are now ready to present the low-sensitivity hopset, given by $\tzspan$ which runs \treeshortcut on the truncated shortest path trees rooted at $v \in V$ whose union over all vertices forms the classic Thorup-Zwick spanner.

\begin{tbox}
    \algname{\tzspan}\\
    \textbf{Input:} A weighted undirected graph $G = (V, E, w)$
    \begin{enumerate}
        \item Let $V = A_0 \supseteq A_1 \supseteq \ldots \supseteq A_{k-1}$ be a hierarchy of vertices chosen according to Theorem~3.7 in \cite{thorup2005approximate}
        \footnote{Theorem~3.7 in \cite{thorup2005approximate} derandomizes the process where each vertex in $A_i$ is independently promoted to $A_{i+1}$ with probability $n^{-1/k}$, and $A_k = \emptyset$.}
        .
        \item $H \gets \emptyset$
        \item For $u \in V$
            \begin{enumerate}
                \item Let $T_u$ be the shortest path tree rooted at $u$, truncated to $T^*_u$ so that $V(T^*_u) = \cluster(u) \cup \set{u}$
                \item $H \gets H \cup \treeshortcut(T^*_u \textrm{ rooted at } u)$
            \end{enumerate}
        \item Return $H$
    \end{enumerate}
\end{tbox}

Let us first see that this construction indeed has low sensitivity, before we come back to the Thorup-Zwick emulator alluded to in the preamble of this section.
Towards this, we define the inverse notion of $\cluster(u)$.

\begin{definition}[\textit{Bunch}]
\label{def:bunch}
    Let $V = A_0 \supseteq A_1 \supseteq \ldots \supseteq A_{k-1}$ be a hierarchy of vertices, and $A_k = \emptyset$.
    The \emph{$i$ level bunch of $v \in V$} is defined as follows:
    \[\bunch_i(v) = \set{u \in A_i : \dist_w(v,u) < \dist_w(v,A_{i+1})}.\]
    The \emph{bunch of $v$} is
    \[\bunch(v) = \bigcup_i \bunch_i(v).\]
\end{definition}

Observe that $v \in \cluster(u)$ if and only if $u \in \bunch(v)$.
We are now ready to show that $\tzspan$ returns a low-sensitivity hopset.
\begin{lemma}
\label{lem:apx-sens}
    $\tzspan$ outputs a hopset $H$ such that $\inorm{\vsens} = O(kn^{1/k}\log^2 n)$.
\end{lemma}
\begin{proof}
    Observe that $\vsens(v)$ for any $v \in V$ is only ever increased by $\tzspan(G)$ in two cases:
    \begin{itemize}
        \item The call to $\treeshortcut(T^*_v \textrm{ rooted at } v)$;
        \item Calls to $\treeshortcut(T^*_u \textrm{ rooted at } u)$ when $v \in T^*_u$ which holds if and only if $u \in \bunch(v)$.
    \end{itemize}
    Each case contributes $O(\log n)$ to $\vsens(v)$, using \Cref{obs:tree-shortcut}.
    Since the first case only occurs once, what remains is to address the second case by bounding $\card{\bunch(v)}$.

    Theorem~3.7 in in \cite{thorup2005approximate} shows that the choice of $A_0, A_1, \ldots, A_k$ results in $\card{\bunch(v)} = O(k n^{1/k} \log n)$ for all $v \in V$.

    Putting everything together, $\inorm{\vsens}$ is bounded above by $\log n \cdot \paren{1 + O(k n^{1/k} \log n)} = O(kn^{1/k}\log^2 n)$.
\end{proof}

Setting $k = \Theta(\sqrt{\log n})$ in \Cref{lem:apx-sens} gives subpolynomial sensitivity: $\inorm{\vsens} = n^{o(1)}$.
Now it remains to show that what $\tzspan$ returns is indeed a good hopset.

\paragraph{Using the Thorup-Zwick as an Analysis Tool. }
We now describe the Thorup-Zwick Emulator of \cite{thorup2006spanners,huang2019thorup}, so that we may presently show that $\tzspan$ is a good hopset.
Let us first make a few more definitions.

\begin{definition}[$p_i(v)$ and its tiebreaking]
    Let $V = A_0 \supseteq A_1 \supseteq \ldots \supseteq A_{k-1}$ be a hierarchy of vertices, and $A_k = \emptyset$.
    For $i \in [0,k-1]$, $p_i(v) \in A_i$ is a vertex such that $\dist_w(v, p_i(v)) = \dist_w(v, A_i)$.
    
    There could be many possible candidates for $p_i(v)$.
    Ties for $p_{k-1}(v)$ are broken arbitrarily.
    If $\dist_w(v, A_i) = \dist_w(v, A_{i+1})$, then we set $p_i(v) = p_{i+1}(v)$; otherwise, we break ties for $p_i(v)$ by setting it to an arbitrary candidate.
    
    Finally, $p_k(v)$ is undefined, but we say that $\dist_w(v, p_k(v)) = \infty$.
\end{definition}

While the tiebreaking is not necessary in the emulator construction in \cite{thorup2006spanners}, it will be useful for our purposes; what it buys us is that we are ensured $p_i(v) \in \bunch(v)$, shown in Lemma~4.1 of \cite{thorup2005approximate} which uses the same tiebreaking scheme.

\begin{definition}[Open and Closed Balls]
\label{def:op-cl-ball}
    Let $V = A_0 \supseteq A_1 \supseteq \ldots \supseteq A_{k-1}$ be a hierarchy of vertices, and $A_k = \emptyset$.
    Let $v \in A_i \setminus A_{i+1}$ for $i \in [0,k-1]$.
    We define the \emph{open ball} of $v$ to be
    \[\ball(v) = \set{u \in V_i : \dist_w(v,u) < \dist_w(v, p_{i+1}(v)}.\]
    The \emph{closed ball} of $v$ is
    \[\ball[v] = \ball(v) \cup \set {p_{i+1}(v)}.\]
\end{definition}

Note in particular that, since we have earlier defined $\dist_w(v, p_k(v)) = \infty$, it follows that for $v \in A_{k-1}$ we have $\ball(v) = \ball[v] = V$.
The emulator construction is then given below by $\tzemu$.

\begin{tbox}
    \algname{\tzemu}\\
    \textbf{Input:} A weighted undirected graph $G = (V, E, w)$
    \begin{enumerate}
        \item Let $V = A_0 \supseteq A_1 \supseteq \ldots \supseteq A_{k-1}$ be a hierarchy of vertices where each vertex in $A_i$ is independently promoted to $A_{i+1}$ with probability $n^{-1/k}$, and set $A_k = \emptyset$
        \item Return $\bigcup\limits_{v \in V} \set{(v,u) : u \in \ball[v]}$
    \end{enumerate}
\end{tbox}

Huang and Pettie showed in \cite{huang2019thorup} the following proposition with regards to $\tzemu$ constructing a hopset.

\begin{proposition}[Theorem~1 in \cite{huang2019thorup}]
\label{prop:huang-pettie}
    Fix any weighted graph $G$ and integer $k \ge 1$.
    $\tzemu$ returns a $(\beta, \eps)$-hopset for $G$ with with size $O(n^{1 + \frac{1}{2^{k+1}-1}})$ and $\beta = 2\paren{\frac{(4+o(1))k}{\eps}}^k = O\paren{(\frac{k}{\eps})^k}$.
\end{proposition}

In order to ascertain that $\tzspan$ gives an $\paren{\Ot\paren{(k/\eps)^k}, \eps}$-hopset, it therefore suffices to show that all edges in $\tzemu$ can be traversed using a polylogarithmic number of edges of $\tzspan$.

\begin{lemma}
\label{lem:tzspan-contains-tzemu}
    For any edge $(v,u)$ in the set returned by $\tzemu$, the set returned by $\tzspan$ contains a path with the same weight as $(v,u)$ and, moreover, has $O(\log^2 n)$ hop-length.
\end{lemma}
\begin{proof}
    Fix any edge $(v,u)$ in the set returned by $\tzemu$, where $v \in A_i \setminus A_{i+1}$ and $u \in \ball[v]$.
    We split into two cases.

    Case 1: $u \in \ball(v)$.
    By comparing definitions, note that $\ball(v) = \bunch_i(v) \subseteq \bunch(v)$.
    Thus $u \in \bunch(v)$ or, equivalently, $v \in \cluster(u)$.

    Case 2: $u = p_{i+1}(v)$.
    By Lemma~4.1 of \cite{thorup2005approximate}, we have $u \in \bunch(v)$ and so $v \in \cluster(u)$.
    
    In either case, by \Cref{obs:clus-trunc} and \Cref{obs:tree-shortcut} used on $T^*_u$ (for the definition see the main loop of $\tzspan$), $H$ contains a shortest path from $u$ to $v$ (and thus $v$ to $u$ since the graph is undirected) using $O(\log^2 n)$ edges.
\end{proof}

And thus we arrive at the main result of this section.

\apxhopsetub*
\begin{proof}
    We show the theorem for \tzspan.
    This follows immediately from combining \Cref{lem:tzspan-contains-tzemu} with \Cref{prop:huang-pettie} to show that $H$ is a $\paren{O\paren{(k/\eps)^k \log^2 n}, \eps}$-hopset, and \Cref{lem:apx-sens} to finish the claim on sensitivity.
\end{proof}

\section{Lower Bounds: Sensitivity-Diameter Tradeoffs}
\label{sec:lb}
In this section we show unconditional lower bounds for how low both the sensitivity and diameter of a construction can simultaneously be.
Namely, we show \Cref{cor:exact-lb,cor:reachability-lb} in \Cref{subsec:lower-hop-shortcut} and \Cref{thm:apx-hopset-lb} in \Cref{subsec:lower-apx-hopset}.

\subsection{Tradeoffs via Perfect Paths}
\label{subsec:lower-hop-shortcut}
Our lower bounds for exact hopsets and shortcut sets go through layered graphs endowed with a set of so-called perfect paths.

\begin{definition}[Perfect Paths, Definition~8 in \cite{bodwin2021new}]
\label{def:perfect-path}
    Let $G=(V_1 \sqcup V_2 \sqcup \ldots \sqcup V_\ell, E)$ be a layered graph.
    A set of paths $\Pi$ is \emph{perfect} if each $\pi \in \Pi$ is the unique shortest path between its endpoints, each $\pi$ starts in $V_1$ and ends in $V_\ell$ with exactly one node in each layer, and each $e \in E$ is in exactly one $\pi \in \Pi$.
\end{definition}
Layered graphs with perfect paths have been used to prove lower bounds for traditional shortcut/hopsets \cite{hesse2003directed,coppersmith2006sparse,huang2021lower,kogan2022having}, and here we repurpose these constructions, with some changes, for our lower bounds.
The ``engine'' of this section is \Cref{thm:perf-path-lb} below, for which we first give some context before proving.

\begin{restatable}{theorem}{perfpathlb}
\label{thm:perf-path-lb}
    If there exists an $\ell$-layered graph $G$ with $n$ nodes in each layer, and a set of perfect paths $\Pi$, then there exists a $2\ell$-layered graph $G'$ with $n$ nodes in each layer such that any $\ell$-shortcut set or $(\ell, 0)$-hopset $H$ on $G'$ must have $\inorm{\esens} \ge \frac{\card{\Pi}}{2n}$. 
\end{restatable}

\paragraph{Simplifying Assumptions. }
First, observe that if $G$ is unweighted and we are concerned with directed $\ell$-shortcut sets, then for any $s$ that reaches $t$, the length of any $st$-path is completely determined by the layers they belong to; any shortcut set on $G$ is thus also a hopset on $G$ (where all edges of $G$ are given weight $1$).
It is therefore sufficient to show the theorem for $(\ell, 0)$-hopsets.

\paragraph{Definition of the graph $G'$.}
We can show a lower bound for vertex sensitivity using $G$ from the hypothesis of \Cref{thm:perf-path-lb} directly (in fact, the proof is simpler).
However, to get a lower bound for \emph{edge} sensitivity, a little more care needs to be taken: we add a $0$ weight edge in the ``middle'' of each vertex of $G$, which gives us $G'$.
We now describe $G'$ more formally.

$G'$ is a standard transformation of $G$ comprising of (i) vertices $v^{in}, v^{out}$ for all $v \in V(G)$; (ii) $0$ weight edges $(v^{in},v^{out})$ for all $v \in V(G)$; (iii) edges $(u^{out}, v^{in})$ for all $(u,v) \in E(G)$, with the same weight.
The proof will only look at the set of paths $\Pi'$ in $G'$ which are ``lifted'' from $\Pi$ (dually, which project down back to $\Pi$), formally defined as follows:
\[\Pi' = \set{(v^{in}_1, v^{out}_1, v^{in}_2, v^{out}_2, \ldots, v^{in}_\ell, v^{out}_\ell) : (v_1, v_2, \ldots, v_\ell) \in \Pi}.\]
Observe that every path $\pi' \in \Pi'$ is the unique shortest path between its endpoints (which starts in the first layer and ends in layer $2\ell$), a property inherited from $\Pi$ being a set of perfect paths.
From this, notice that there is no choice in selecting the routing path from $s$ to $t$ for any $s,t \in \pi'$ for all $\pi' \in \Pi'$; any such routing path must be the subpath of $\pi'$ from $s$ to $t$ since $\pi'$ is the unique shortest path between its endpoints, and so the span of hopset edges $(s,t)$ are fixed.
Moreover, each edge of the form $(u_{out}, v_{in})$ is in exactly one $\pi' \in \Pi'$ since $\Pi$ is perfect.
Finally, note that $\card{\Pi'} = \card{\Pi}$.
We are now ready to state the proof.

\begin{proof}[Proof of \Cref{thm:perf-path-lb}]
    Let $G'$ and $\Pi'$ be defined as above, and let $H$ be any $(\ell,0)$-hopset for $G'$.
    Let $H' \subseteq H$ be a minimal set of edges which ensures that there are $\ell$ hop-length shortest paths in $G' \cup H'$ between the endpoints of $\pi'$ for all $\pi' \in \Pi'$.
    Observe that $\inorm{\esens_{G',H}} \ge \inorm{\esens_{G',H'}}$ since $H' \subseteq H$ and so it suffices to lower bound the latter which we now write without subscripts.

    Let $T$ be the set of all edges of the form $(v_{in}, v_{out})$ and $\mathds{1}_{T}$ be its indicator vector; note for later that $\card{T} = \onenorm{\mathds{1}_{T}} = n \ell$.
    We will lower bound the total sensitivity of edges in $T$ and use an averaging argument to get a bound on $\inorm{\esens}$.
    More specifically, we will show \[\inorm{\esens} \ge \mathds{1}_{T} \cdot \esens / \onenorm{\mathds{1}_{T}} \ge x/2\] where $\inorm{\esens} \ge \esens \cdot \mathds{1}_{T} / \onenorm{\mathds{1}_{T}}$ comes from a simple averaging argument\footnote{To be overly indulgent, from an application of H{\"o}lder's inequality.}.
    To this end, consider the process where we add the edges of $H' = \set{e_1, e_2, \ldots, e_{\card{H'}}}$ to $G'$ one by one so that $H'_0 = \set{}$ and $H'_i = H'_{i-1} \cup \set{e_i}$, and define the potential \[\Phi_i = \sum_{\pi' \in \Pi'} \card{E(P_{G' \cup H'_i}(\pi'))}\] where $P_{G' \cup H'_i}(\pi')$ denotes a shortest path in $G' \cup H'_i$ with the smallest hop-length between the endpoints of $\pi'$.
    $\Phi_i$ is then the sum of said hop-lengths running over $\pi' \in \Pi'$ at stage $i$.
    
    Observe that the initial value $\Phi_0$ is $2\card{\Pi}\ell$ (since there are $\card{\Pi} = \card{\Pi'}$ paths in $\Pi'$, each with hop-length exactly $2 \ell$), and the final value $\Phi_{\card{H'}}$ is at most $\card{\Pi}\ell$ (since $H'$ is an $(\ell,0)$-hopset).
    Moreover, at most one summand of $\Phi_i$ differs from the corresponding summand of $\Phi_{i-1}$ for the following reason: $\edgespan(e_i)$ is contained entirely and only in one $\pi' \in \Pi'$ by the minimality of $H'$ and uniqueness of $\pi'$.
    For the same reason, $\edgespan(e_i)$ alternates between edges not in $T$ and edges in $T$ and so twice the contribution of $e_i$ to $\mathds{1}_{T} \cdot \esens$ is at least the decrease in potential $\Phi_{i-1} - \Phi_i$.
    That is, \[\Phi_i + 2\mathds{1}_{T} \cdot \esens_{G',H'_i} > \Phi_{i-1} + 2\mathds{1}_{T} \cdot \esens_{G',H'_{i-1}}.\]
    Using our observation about the initial and final values of $\Phi$, we conclude that
    \begin{align*}
        &
        \card{\Pi}\ell + 2\mathds{1}_{T} \cdot \esens \ge \Phi_{\card{H'}} + 2\mathds{1}_{T} \cdot \esens > \Phi_0 + 2\mathds{1}_{T} \cdot \esens_{G',H'_0} = 2\card{\Pi}\ell
        \\
        \implies
        &
        \mathds{1}_{T} \cdot \esens \ge \frac{\card{\Pi}\ell}{2}
        \\
        \implies
        &
        \frac{\mathds{1}_{T} \cdot \esens}{\onenorm{\mathds{1}_T}} \ge \frac{\card{\Pi}}{2n}. \tag{$\onenorm{\mathds{1}_T} = n\ell$}
    \end{align*}
    We conclude from the last line that $\inorm{\esens} \ge \frac{\card{\Pi}}{2n}$.
\end{proof}

Finally, we can invoke \Cref{thm:perf-path-lb} with known constructions to get unconditional lower bounds on the sensitivity-diameter tradeoff.

\begin{proposition}[Theorem~5 of \cite{bodwin2021new}, earlier versions of which appear in 
\cite{coppersmith2006sparse,szemeredi1983extremal}]
\label{prop:bod-weighted-perf-paths}
    For any integers $n, \ell \le n$, and $x \le n/\ell$, there is an $\ell$-layered weighted graph $G$ with $n$ nodes in each layer and a set of perfect paths $\Pi$ such that each node is in exactly $x$ paths in $\Pi$.
\end{proposition}

\exactlb*
\begin{proof}
    For large enough $N$, we invoke \Cref{prop:bod-weighted-perf-paths}, setting
    \begin{itemize}
        \item $n \gets \frac{N}{2\beta}$
        \item $\ell \gets 2\beta$
        \item $x \gets \frac{N}{4\beta^2}$;
    \end{itemize}
    call this graph $G_N$.
    Note that $\card{\Pi}/(2n) = x/2$.
    We then pass each $G_N$ into \Cref{thm:perf-path-lb} to get
    \[\inorm{\esens} \ge \frac{x}{2} = \frac{N}{8\beta^2} \implies \inorm{\esens} \cdot \beta^2 \ge \frac{N}{8},\]
    which completes the proof.
\end{proof}

Observe that \Cref{cor:exact-lb} shows that \Cref{thm:greedy-good} is tight up to polylogarithmic factors, but it remains open whether we can explicitly trade between the sensitivity and diameter.
For example, can we find a $(\Ot(n^{1/2-c}),0)$-hopset with $\inorm{\vsens} = \Ot(n^{2c})$ for any $c$?
We close the discussion on exact hopsets with a remark on a connection between the problem of packing perfect paths in layered graphs and a possible threshold effect in the upper bounds for sensitivity-diameter tradeoffs.

\begin{remark}\label{rem:perfect}
    Observe that \Cref{thm:perf-path-lb} connects the existence of layered graphs with many perfect paths to sensitivity-diameter tradeoffs, and this concretely manifests as a $\inorm{\esens}\beta^2 = \Omega(N)$ bound by using the construction of \Cref{prop:bod-weighted-perf-paths}.
    What are the ramifications of being able to pack even more perfect paths into an $\ell$-layered graph?
    Let us explore one particular hypothetical scenario.
    Suppose it is possible to pack $N$ perfect paths into $\sqrt{N}$-layered graphs.
    Then, we are able to show (in exactly the same way as \Cref{cor:exact-lb}, but with different parameters) an $\inorm{\esens}\beta = \Omega(N)$ bound for $\beta \le \sqrt{N}/2$.
    It is conceivable that such a construction may exist\footnote{And such a construction does not contradict \Cref{thm:greedy-good} since, there, $\beta \ge \sqrt{N}\log n > \sqrt{N}/2$.}.
    If such a packing does indeed exist for infinitely many values of $N$, this would rule out the possibility of achieving a smooth sensitivity-diameter tradeoff on the upper bound side.
    By \Cref{thm:greedy-good}, \greedyshortcut produces a $(\Ot(\sqrt{n}),0)$-hopset with $O(\log n)$ sensitivity, and any polynomial improvement to $\beta$ would raise the sensitivity very abruptly from $O(\log n)$ to $\Omega(\sqrt{n})$.
    That is to say, a $(\Ot(n^{1/2-c}),0)$-hopset construction could hope for no better than $\inorm{\vsens} = \widetilde{\Omega}(n^{1/2+c})$ if we can pack $N$ perfect paths into $\sqrt{N}$-layered graphs.
    On the flip side, this means that achieving a tradeoff like a $(\Ot(n^{1/2-c}),0)$-hopset construction with $\inorm{\vsens}$ less than $\Ot(n^{1/2+c})$ by a polynomial factor, for some constant $c > 0$, would rule out the existence of $\sqrt{N}$-layered graphs packed with $N$ perfect paths.
    A similar line of reasoning also shows that any polynomial improvement to $\beta$ (keeping $\inorm{\vsens}$ = $\Ot(\sqrt{n})$) over \folkloreshortcut will rule out the existence of $\sqrt{N}$-layered graphs packed with $N$ perfect paths.
    The existence of layered graphs packed with as many perfect paths as possible is an open problem related to many lower bounds in the hopset/distance preservers/etc literature \cite{bodwin2023folklore,coppersmith2006sparse}; particularly, we would get more streamlined lower bounds for hopsets matching the result of \cite{bodwin2023folklore}.
\end{remark}

We next show a lower bound for reachability in directed graphs.

\begin{proposition}[Theorem~6 of \cite{bodwin2021new}, earlier versions of which appear in \cite{abboud20174,alon2002testing,behrend1946sets,coppersmith2006sparse,huang2021lower}]
\label{prop:bod-unweighted-perf-paths}
    For any integers $n, d \ge 2, \ell \le n^{1/d}$, and
    \[
        x = O\Paren{n^{\frac{d-1}{d+1}}\ell^{-d\frac{d-1}{d+1}}},
    \]
    there is an $\ell$-layered unweighted graph $G$ with $n$ nodes in each layer and a set of perfect paths $\Pi$ such that each node is in exactly $x$ paths in $\Pi$.
\end{proposition}

\reachabilitylb*
\begin{proof}
    For large enough $N$, we invoke \Cref{prop:bod-unweighted-perf-paths}, setting
    \begin{itemize}
        \item $n \gets \frac{N}{2\beta}$
        \item $d \gets 2$
        \item $\ell \gets 2\beta$
        \item $x \gets \Theta(\frac{N^{1/3}}{\beta})$;
    \end{itemize}
    call this graph $G_N$.
    Note that $\card{\Pi}/(2n) = x/2$.
    We then pass $G_N$ into \Cref{thm:perf-path-lb} to get
    \[\inorm{\esens} \ge \frac{x}{2} = \Theta(\frac{N^{1/3}}{\beta}) \implies \inorm{\esens} \cdot \beta \ge \Omega(N^{1/3}),\]
    which completes the proof.
\end{proof}

Curiously, none of our upper bounds and lower bounds for directed shortcut/hopsets match each other.
New ideas are probably needed to better understand sensitivity-diameter tradeoffs (for both shortcut sets and hopsets) in the directed setting.

\subsection{Tradeoffs for Approximate Hopsets}
\label{subsec:lower-apx-hopset}
In this section we give a tradeoff lower bound between $\eps, \beta$, and $\inorm{\esens}$ for approximate hopsets.
While the sensitivity of a hopset is not equal to its size, we can convert some claims pertaining to the size of a hopset to a slightly weaker claim about the sensitivity of a hopset.
\cite{abboud2018hierarchy} gives a tradeoff between $\eps, \beta$, and the size of a hopset $\card{H}$, which we use directly to get the aforementioned type of sensitivity tradeoff.

\begin{proposition}[Theorem~4.6 in \cite{abboud2018hierarchy}]
\label{prop:abboud-bodwin-pettie}
    Fix a positive integer $k$ and parameter $\eps > 1/n^{o(1)}$.
    Any construction of $(\beta, \eps)$-hopsets with size less than $n^{1 + \frac{1}{2^k - 1} - \Delta}$, $\Delta > 0$, has $\beta = \Omega\paren{\paren{\frac{1}{2^{k-2}(2k-1)\eps}}^k}$.
\end{proposition}

Using this, we get the following observation as a warmup.

\begin{observation}
\label{obs:apx-hopset-vert-lb}
    Fix a positive integer $k$ and parameter $\eps > 1/n^{o(1)}$.
    Any construction of $(\beta, \eps)$-hopsets $H$ with $\beta = O\paren{\paren{\frac{1}{2^{k-2}(2k-1)\eps}}^k}$ has $\inorm{\vsens} \ge n^{\frac{1}{2^k - 1} - \Delta}$, $\Delta > 0$.
\end{observation}
\begin{proof}
    Observe that any edge in $H$ contributes at least $1$ to $\onenorm{\vsens}$.
    Consequently, a hopset with size $n^{1 + \frac{1}{2^k - 1} - \Delta}$ contributes a total of at least $n^{1 + \frac{1}{2^k - 1} - \Delta}$ to $\onenorm{\vsens}$.
    The average vertex sensitivity is then $n^{\frac{1}{2^k - 1} - \Delta}$.
    Since at least one vertex must incur at least the average, it follows that $\inorm{\vsens} \ge n^{\frac{1}{2^k - 1} - \Delta}$.
    
    Taking the contrapositive of \Cref{prop:abboud-bodwin-pettie}, any hopset with $\beta = O\paren{\paren{\frac{1}{2^{k-2}(2k-1)\eps}}^k}$ must have size at least $n^{1 + \frac{1}{2^k - 1} - \Delta}$ which, by our earlier discussion, must have $\inorm{\vsens} \ge n^{\frac{1}{2^k - 1} - \Delta}$.
\end{proof}

\Cref{obs:apx-hopset-vert-lb} is not quite as strong as what we would like to prove; it bounds $\inorm{\vsens}$ but we would like to get a bound on $\inorm{\esens}$.
Averaging over the number of edges does not go through since the graph from \cite{abboud2018hierarchy} contains a super-linear number of edges.
To this end, we have to peek into the guts of the proof of \Cref{prop:abboud-bodwin-pettie} just a little bit.

\apxhopsetlb*
\begin{proof}
    First, we take the recursive construction $G$ used in the proof of \Cref{prop:abboud-bodwin-pettie} and do a standard graph transformation to get $G'$ comprising of (i) vertices $v^{in}, v^{out}$ for all $v \in V(G)$; (ii) $0$ weight edges $(v^{in},v^{out})$ for all $v \in V(G)$; (iii) edges $(u^{out}, v^{in})$ for all $(u,v) \in E(G)$, with the same weight.
    Let $T$ be the set of all edges of the form $(v_{in}, v_{out})$ and $\mathds{1}_{T}$ be its indicator; note for later that $\card{T} = \onenorm{\mathds{1}_{T}} = n$.
    Our goal is as follows: for any $(\beta, \eps)$-hopset $H'$ of $G'$ satisying the hypothesis, we seek to lower bound $\esens \cdot \mathds{1}_{T}$, the sum of sensitivies over edges in $T$.
    
    Inspecting the result of \cite{abboud2018hierarchy}, the proof of \Cref{prop:abboud-bodwin-pettie} actually gives a lower bound on the size of a subset of edges in a hopset with $\beta = O\paren{\paren{\frac{1}{2^{k-2}(2k-1)\eps}}^k}$; these are the so-called \emph{long and owned} edges.
    Long and owned edges $e$ have the following property: $\edgespan(e)$ goes \emph{through} a layer of $G$.
    This means that such an edge ``lifted'' into $G'$ must contain an edge of the form $(v_{in},v_{out})$ in $\edgespan(e')$ and therefore contribute at least $1$ to $\esens \cdot \mathds{1}_{T}$.

    Taking the contrapositive of \Cref{prop:abboud-bodwin-pettie}, any hopset with $\beta = O\paren{\paren{\frac{1}{2^{k-2}(2k-1)\eps}}^k}$ must have size at least $n^{1 + \frac{1}{2^k - 1} - \Delta}$ which, by our earlier discussion, must mean that $\esens \cdot \mathds{1}_{T} \ge n^{1 + \frac{1}{2^k - 1} - \Delta}$.
    By an averaging argument, we get
    \begin{align*}
        \inorm{\esens} \ge \frac{\esens \cdot \mathds{1}_{T}}{\onenorm{\mathds{1}_{T}}} \ge n^{\frac{1}{2^k - 1} - \Delta},
    \end{align*}
    completing the proof.
\end{proof}

Observe that \Cref{thm:apx-hopset-lb} shows that, in the regime where $\eps$ is exactly inverse polylogarithmic, having a polylogarithmic $\beta$ (which is of the form $\beta = O_k\paren{(1/\eps)^k}$ for some constant $k$) necessitates at least a polynomial edge sensitivity of roughly $n^{1/2^k}$.
That is to say, at least one of $\beta$ or $\inorm{\esens}$ needs to be superpolylogarithmic when $\eps$ is inverse polylogarithmic.
We leave open whether we can say something similar for other ranges of $\eps$ (particularly, for $\eps$ constant) which, if true, would show that our upper bound of $(n^{o(1)},\eps)$-hopsets with $\inorm{\vsens} = n^{o(1)}$ is not too far from the best one could hope for: we could hope to improve one of $\beta$ or $\inorm{\vsens}$ to polylogarithmic, but not both simultaneously.
Finally, note that a similar gap is present in the traditional hopset literature: it is not known whether for a constant $\eps > 0$ if we can get a $(O(\log^{k_1} n), 0)$-hopset with $O(n\log^{k_2} n)$ size, for constants $k_1, k_2$.
In fact, the lower bound we have here matches the best known traditional hopset lower bound \cite{abboud2018hierarchy}, where a size $s$ there corresponds to a sensitivity $s/n$ here.

\section{Open Problems}
\label{sec:open}
In this section we collect a list of problems left open by our results.
First and foremost, the picture pertaining to \emph{directed} graphs is least clear.

\begin{itemize}
    \item The bounds for shortcut sets (see \Cref{thm:direted-reach-ub} and \Cref{cor:reachability-lb}) on directed graphs are not tight.
        Can we provide tight results for any specific value of sensitivity (say $O(\log n)$) if not for all values?
    \item The best known upper bound for directed exact hopsets is based on the folklore algorithm for traditional hopsets.
        It had been a long standing open problem whether this was the best we could do, and it is only within the last year that the folklore algorithm has been shown to be optimal \cite{bodwin2023folklore}.
        Can similarly fresh ideas close the gap between the folklore algorithm (see \Cref{lem:folklore}) and lower bounds here (see \Cref{cor:exact-lb})?
    \item We know essentially nothing insightful about the situation for directed approximate hopsets, besides the fact that our shortcut set lower bound in \Cref{cor:reachability-lb} applies (since a directed approximate hopset is a directed shortcut set).
        Find a non-trivial upper bound.
\end{itemize}
While more is known for undirected graphs, there are certainly gaps left to fill.
\begin{itemize}
    \item The upper bound given by \greedyshortcut only works for a specific sensitivity regime (i.e. it completes the picture for only one point on the tradeoff curve).
        Find a non-trivial algorithm, that outputs hopsets close to the curve of \Cref{cor:exact-lb}, where the sensitivity is tunable to values larger than polylogarithmic.
        See \Cref{rem:perfect} for more context surrounding this problem.
    \item Our lower bounds for undirected approximate hopsets in \Cref{thm:apx-hopset-lb} are not tight; they do not address the regime when $\eps$ is a constant well.
        For a constant $\eps > 0$, do there exist $(\beta, \eps)$-hopsets with $\beta$ and $\inorm{\vsens}$ simultaneously polylogarithmic, or must one of the quantities necessarily be superpolylogarithmic?
\end{itemize}
Finally, we believe that hopset sensitivity is a natural notion which should have concrete algorithmic applications beyond differential privacy.
Where else can low-sensitivity hopsets can be used as a primitive?

\section*{Acknowledgements}
The authors would like to thank Greg Bodwin for helpful discussions about hopsets.

\clearpage

\bibliographystyle{alpha}
\bibliography{reference}

\clearpage

\appendix
\section{On Applying \cite{bodwin2023folklore} to Sensitivity-Diameter Lower Bounds}
\label{app:optimal-hopset-lb}
Here we continue \Cref{rem:optimal-hopset-lb} in a little more detail.

A recent breakthrough of Bodwin and Hoppenworth \cite{bodwin2023folklore} was able to prove tight lower bounds for the \emph{number of edges} in a directed exact hopset by using a different kind of construction that relaxes the property that perfect paths (see \Cref{def:perfect-path}) are internally disjoint.
Unfortunately, their relaxation does not work for lower bounding the \textit{sensitivity} of a hopset. 

Loosely speaking, in the construction of Bodwin and Hoppenworth, perfect paths are allowed to intersect, but their pairwise intersection is always a segment with a small number of edges.
This means that a shortcut edge $(s,t)$ can shortcut multiple perfect paths only if the shortcut is low-hop -- i.e. only if $s$ and $t$ are in nearby layers.
Since low-hop shortcuts only make limited progress in reducing the hop-distance between the first and last layer of the overall graph, the authors are able to conclude that it would require a large number of low-hop shortcuts to reduce the number of hops on every perfect path.
But in our setting we do not limit the number of shortcut edges, so by using a large number of these low-hop shortcuts, one can reduce the diameter of their construction while incurring low sensitivity; for this reason, their approach does not lead to an improved lower bound in our setting.

\section{Proofs of \pathshortcut and \treeshortcut Guarantees}
\label{app:path-tree}
For completeness, we prove the guarantees of \pathshortcut and \treeshortcut here.

\obspathshortcut*
\begin{proof}\ \\
    \textbf{Output is a $O(\log n)$-shortcut/hopset.}
    Let $H = \pathshortcut(P)$ and $u$ be any vertex along $P$.
    By an induction argument, there is always an edge $(v_0, v) \in H$ where $v_0$ reaches $v$ which reaches $u$, and the $P$ hop-length from $v_0$ to $v$ is at least as much as that from $v$ to $u$.
    Thus, the $P \cup H$ hop-length from $v_0$ to $u$ is at most $O(\log n)$.
    A similar argument shows that the $P \cup H$ hop-length from $u$ to $v_n$ is at most $O(\log n)$.
    Finally, let $u$ and $v$ be any two vertices in $P$ where $u$ reaches $v$.
    Consider the first time $u$ and $v$ are separated into different recursive instances with endpoints $x, y$ and $y, z$ respectively.
    By the above reasoning, the $P \cup H$ hop-lengths from $u$ to $y$ and from $y$ to $v$ are both at most $O(\log n)$, and by concatenating them we can see that the $P \cup H$ hop-length from $u$ to $v$ is at most $O(\log n)$. 
    \\
    \textbf{Sensitivity is bounded by $O(\log n)$.}
    Consider the recursion tree of \pathshortcut, where each node is labelled with its recursive input.
    For example, the root node is labelled with $P = v_0, v_1, \ldots, v_n$, its left child is labelled with $v_0, v_1, \ldots, v_{n/2}$ and its right child is labelled with $v_{n/2}, v_{n/2+1}, \ldots, v_n$, and so on.
    A vertex is contained in at most two labels per level.
    Since each appearance of a vertex $v$ in a label contributes $1$ to $\vsens(v)$ and there are $O(\log n)$ levels, $\vsens(v) = O(\log n)$ for all $v \in V$ and therefore $\inorm{\vsens} = O(\log n)$.
\end{proof}

\obstreeshortcut*
\begin{proof}\ \\
  \textbf{Output is a $O(\log^2 n)$-shortcut/hopset. }
    Let $H = \treeshortcut(T \textrm{ rooted at } v)$.
    For any pair of vertices $s, t$, consider the unique path in $T$, which by \Cref{prop:hl-decom} is comprised of at most $O(\log n)$ heavy (sub)paths and $O(\log n)$ light edges.
    For each heavy path $P$, any pair of vertices in $P$ have a $P \cup H$ hop-length of $O(\log n)$ by \Cref{obs:path-shortcut}.
    We can stitch together light edges and the aforementioned $O(\log n)$ hop-length paths for each heavy path to show the $T \cup H$ hop-length from $s$ to $t$ is at most $O(\log^2 n)$.
    \\
    \textbf{Sensitivity is bounded by $O(\log n)$. }
    For any heavy path $P$, \Cref{obs:path-shortcut} asserts that $\pathshortcut(P)$ contributes $O(\log n)$ to $\vsens(v)$ for all $v \in V(P)$.
    By \Cref{prop:hl-decom}, heavy paths are vertex disjoint and, consequently, the only contribution to $\vsens(v)$ is from the call to $\pathshortcut(P)$ where $v \in V(P)$.
    Thus, $\inorm{\vsens} = O(\log n)$.
\end{proof}
\section{Undirected Shortcut Sets Upper Bound}
\label{app:warmup-1}
A direct application of \treeshortcut gives us undirected low sensitivity shortcut sets.

\begin{tbox}
    \algname{\ezshortcut}\\
    \textbf{Input:} An undirected graph $G = (V, E)$
    \begin{enumerate}
        \item $H \gets \emptyset$
        \item For each connected component $C$ of $G$
            \begin{enumerate}
                \item Select a representative $c \in V(C)$
                \item $H \gets H \cup \treeshortcut(\textrm{the union of all routing paths rooted at } c)$
            \end{enumerate}
        \item Return $H$
    \end{enumerate}
\end{tbox}

\begin{observation}
\label{obs:ez-shortcut}
    \ezshortcut produces an $O(\log^2 n)$-shortcut set with $\inorm{\vsens} = O(\log n)$.
\end{observation}
\begin{proof}
    This follows immediately from \Cref{obs:tree-shortcut} and observing that reachability is preserved.
\end{proof}
\section{Directed Exact Hopsets Upper Bound}
\label{app:warmup-2}
Using ideas from \cite{ghazi2022differentially,deng2023differentially}, we can modify a folklore randomized algorithm (attributed to \cite{ullman1990high}), which gives a directed $(\Ot(\sqrt{n}),0)$-hopset with linear size, to get a directed $(\Ot(\sqrt{n}),0)$-hopset with $\inorm{\vsens} = \Ot(\sqrt{n})$.

\begin{tbox}
    \algname{\folkloreshortcut}\\
    \textbf{Input:} A directed graph $G = (V, E, w)$
    \begin{enumerate}
        \item Take each $v \in V$ into the set $X$ independently with probability $n^{-1/2}$
        \item $H \gets \emptyset$
        \item For each $v \in X$
            \begin{enumerate}
                \item $H \gets H \cup \treeshortcut(\textrm{shortest path arborescence rooted at } v)$
            \end{enumerate}
        \item Return $H$
    \end{enumerate}
\end{tbox}

\begin{lemma}
\label{lem:folklore}
    $\folkloreshortcut$ produces a directed $(O(\sqrt{n} \log n),0)$-hopset with $\expect{\inorm{\vsens}} = O(\sqrt{n} \log n)$.
\end{lemma}
\begin{proof}\ \\
    \textbf{With high probability, all shortest path hop-lengths are bounded by $(O(\sqrt{n} \log n)$. }
    Let $P$ be a $G$ shortest path from $u$ to $v$ with $\Omega(\sqrt{n} \log n)$ hop-length.
    Then, with high probability at least one of the first $\Theta(\sqrt{n} \log n)$ vertices of $P$ is sampled into $X$; call any one of these vertices $x$ and let its shortest path arborescence be denoted with $T_x$.
    By \Cref{obs:tree-shortcut}, $\treeshortcut$ gives an $O(\log^2 n)$ hop-length $G \cup T_x$ shortest path from $x$ to $v$ and so a shortest path from $u$ to $v$ can be taken using $O(\sqrt{n} \log n)$ edges in $G$ from $u$ to $x$ and $O(\log^2 n)$ edges from $x$ to $v$ using edges in $G \cup T_x$, which totals to $(O(\sqrt{n} \log n)$ edges.
    By a union bound, this holds with high probability for all $O(n^2)$ shortest paths using $\Omega(\sqrt{n} \log n)$ edges.
    \\
    \textbf{Sensitivity is bounded by $O(\sqrt{n} \log n)$ in expectation. }
    The expected size of $X$ is $\sqrt{n}$, and by \Cref{obs:tree-shortcut} each call to $\treeshortcut(T_x)$ for $x \in X$ contributes at most $O(\log n)$ to $\vsens(v)$ for all $v \in V$, hence the claim follows.
\end{proof}

A simple application of the probabilistic method on \Cref{lem:folklore} then shows the existence of $(O(\sqrt{n} \log n),0)$-hopsets with $\inorm{\vsens} = O(\sqrt{n} \log n)$.
\section{Applications to Differential Privacy}
\label{app:dp-rq-sd}
In this section, we present the application of low-sensitivity hopsets in the problem of All Sets Range Queries (ASRQ) with differnetial privacy introduced in \Cref{subsec:app-to-dp}. We first show the proof of \Cref{thm:pure-dp} for the ASRQ problem and discuss the connection between low-sensitivity hopsets and All Pairs Shortest Distances (APSD) problem. 

\subsection{Technical Tools}
The main technical tool we use for the proof is the Laplace Mechanism. For completeness we provide the necessary definitions.

\begin{definition}[\textit{Laplace distribution}]
\label{def:lap-dist}
\normalfont We say a zero-mean random variable $X$ follows the Laplace distribution with parameter $b$ (denoted by $X \sim \Lap{b}$) if the probability density function of $X$ follows
\begin{align*}
p(x) = \Lap{b}\paren{x} = \frac{1}{2b}\cdot \exp\paren{-\frac{\card{x}}{b}}.
\end{align*}
\end{definition}

Laplace random variables exhibit a nice concentration property, formally as follows.

\begin{lemma}[Sum of Laplace random variables,~\cite{chan2011private,wainwright2019high}]
\label{lem:lap-sum}
\normalfont Let $\{X_{i}\}_{i=1}^{m}$ be a collection of independent random variables such that  $X_i \sim \Lap{b_{i}}$ for all $1\leq i \leq m$. Then, for $\nu \geq \sqrt{\sum_{i} b^2_{i}}$ and $0<\lambda<\frac{2\sqrt{2}\nu^2}{b}$ for $b= \max_{i} \set{b_{i}}$, 
\begin{align*}
\prob{\abs{\sum_{i} X_{i}} \geq \lambda}\leq 2\cdot \exp\paren{-\frac{\lambda^2}{8\nu^2}}.
\end{align*}
\end{lemma}

We move forward to sensitivity as known in DP, then we are able to define the Laplace mechanism, a standard DP mechanism that adds noise sampled from Laplace distribution with scale dependent on the $\ell_1$-sensitivity.

\begin{definition}[\textit{Sensitivity}]
\label{def:sensitivity}
\normalfont Let $p \geq 1$. For any function $f: \mathcal{X}\rightarrow \mathbb{R}^k$ defined over a domain space $\mathcal X$,  the $\ell_p$-sensitivity of the function $f$ is defined as  
\begin{displaymath}
\Delta_{f,p} = \max_{\substack{w,w' \in \mathcal{X} \\
w\sim w'}}\|f(w)-f(w') \|_{p},
\end{displaymath}
Here, $\norm{\bx}_p:=\paren{\sum_{i=1}^d \abs{\bx[i]}^p }^{1/p}$ is the $\ell_p$-norm of the vector $\mathbf x \in \mathbb R^d$ and $\bx[i]$ is the $i$-th coordinate.
\end{definition}

\begin{definition}[\textit{Laplace mechanism}]
\label{def:lap-mech}
\normalfont
For any function $f: \mathcal{X}\rightarrow \mathbb{R}^k$, the Laplace mechanism on input $w\in \mathcal{X}$ samples $Y_1, \dots, Y_k$ independently from $\Lap{\frac{\Delta_{f,1}}{\varepsilon}}$ and outputs
\begin{displaymath}
M_{\varepsilon}(f) = f(w) + (Y_1, \dots, Y_k).
\end{displaymath}
\end{definition}

The last piece we need is the basic composition theorem.

\begin{proposition}[Composition theorem \cite{dwork2006calibrating}]
\label{prop:basic-comp}
For any $\varepsilon > 0$, the composition of $k$ $\varepsilon$-differentially private algorithms is $k\varepsilon$-differentially private. 
\end{proposition}

\subsection{Proof of \Cref{thm:pure-dp}}
\label{subsec:proof-of-dp}
For completeness we define the DP-ASRQ formally.
\begin{definition}[\textit{Range Queries with Neighboring Attributes}]
\label{def:neighbor-attributes}
\normalfont
Let $(\mathcal{R}=(X, \mathcal{S}), f)$ be a system of range queries, and let $w, w': X \rightarrow \mathbb{R}^{\geq 0} $ be attribute functions that map each element in $X$ to a non-negative real number. we say $w, w'$ are neighboring, denoted as $w \sim w'$ if $\sum_{e\in E}|w(e)-w'(e)| \leq 1.$
\end{definition}

\begin{definition}[\textit{Differentially Private Range Queries}]
\label{def:dp-asrq}
\normalfont
Let $(\mathcal{R}=(X, \mathcal{S}), f)$ be a system of range queries and $w, w': X \rightarrow \mathrm{R}^{\geq 0} $ be neighboring attribute functions. Let $\mathcal{A}$ be an algorithm that takes $(\mathcal{R}, f , w)$ as input. Then $\mathcal{A}$ is $(\epsilon, \delta)$-differentially private on $\mathcal{G} $ if, for neighboring attribute functions $w\sim w'$ and all sets of possible outputs $\mathcal{C}$, we have: $\Pr[\mathcal{A}(\mathcal{R}, f , w)\in \mathcal{C}] \leq e^{\epsilon}\cdot \Pr[\mathcal{A}(\mathcal{R}, f , w') \in \mathcal{C} ]+\delta$.
\end{definition}

We show an algorithm named as \puredpalg that calls the $(O(\sqrt{n}\log n), 0)$-hopset construction with $O(\log n)$ sensitivity as a subroutine. The high-level idea is to apply \greedyshortcut to the given graph $G$ and construct $H$ using the true attributes. Now each hopset edge in $H$ has an attribute with the true summation. To guarantee DP we add Laplace noise to the attributes on hopset edges, the noise magnitude is tempered by the sensitivity (as defined in DP), which is the hopset edge sensitivity (as defined in this paper). In addition we also add Laplace noise for each edge in the original graph. For a fixed shortest path, we report the query output along a 
path of at most $O(\sqrt{n}\log n)$ hops with edges in $E\cup H$, adding up the perturbed attributes of the edges. The additive error is at most the summation of all noises added. By concentration of Laplace random variables we get the claimed bound.

\begin{tbox}
    \textbf{\algname{\puredpalg}: An $\epsilon$-DP algorithm to release all sets range queries}\\
    \textbf{Input:} Undirected graph $G = (V, E, w)$, shortest paths $\cP$ and privacy parameter $\epsilon>0$
    \begin{enumerate}
        \item $H \gets \greedyshortcut(G, \cP)$
        \item Let the new graph be $G' = (V, E+H, w_H)$, and shortest paths using hop edges be $\cP'$
        \item For each hop edge $e \in H$
            \begin{enumerate}
                \item Add an independent Laplace noise $\Lap{2\log n/\epsilon}$ to $w_H(e)$, denoted as $f'(e)$
            \end{enumerate}
        \item For each edge $e \in E$
            \begin{enumerate}
                \item Add an independent Laplace noise $\Lap{2/\epsilon}$ to $w_H(e)$, denoted as $f'(e)$
            \end{enumerate}
        \item For each $(u,v) \in V\times V$
            \begin{enumerate}
                \item Report the query output as $\hat{f}(u,v) = \sum_{e \in E(P_{uv})}f'(e)$
            \end{enumerate}
    \end{enumerate}
\end{tbox}
The following lemma immediately proves \Cref{thm:pure-dp}.

\begin{lemma}
    \label{lem:hop-asrq}
    For any undirected graph $G$ and any $\epsilon > 0$, \puredpalg is an $\epsilon$-DP algorithm that outputs all sets range queries with additive error at most $\Ot(\frac{n^{1/4}}{\epsilon})$ with high probability.
\end{lemma}
\begin{proof}\ \\
    \textbf{$\epsilon$-DP guarantee.} The privacy proof directly comes from the Laplace mechanism, which is applied twice in \puredpalg. By \Cref{thm:greedy-good} the sensitivity of each hop edge $e \in H$ is $\log n$, and the sensitivity of individual edges in $E$ is simply 1. Therefore by \Cref{def:lap-mech} Step 3 and 4 in \puredpalg both satisfy $\epsilon/2$-DP. By the composition theorem \Cref{prop:basic-comp}, we arrive at $\epsilon$-DP for \puredpalg.
    \\
    \textbf{Analysis of additive error.} Let $H$ be the hospet given by \puredpalg, the new graph $G' = G\cup H$ now has at most $\Ot(\sqrt{n})$ hop diameter. For edges any shortest paths $P'_{uv}$ in $G'$, there could be hop edges with large perturbation from step 3, and edges with small perturbation from step 4. However, the number of edges involved on a shortest path is always bounded by $\Ot(\sqrt{n})$. Therefore the maximum additive error can be decomposed into $\Ot(\sqrt{n})$ independent noises sampled from $\Lap{2\log n/\epsilon}$. Applying \Cref{lem:lap-sum} with a standard union bound argument, we get the upper bound of $\Ot(\frac{n^{1/4}}{\epsilon})$.
\end{proof}

Recall prior results in~\cite{deng2023differentially}, we are able to improve the $\epsilon$-DP additive error upper bound from $\Ot(n^{1/3})$ to $\Ot(n^{1/4})$, matching the $\Ot(n^{1/4})$ upper bound for the $(\epsilon, \delta)$-DP setting, which is tight up to polylogarithmic factors according to a recent work~\cite{bodwin2024discrepancy}. Another remark is that other exact hopset constructions can also be applied to the framework of \puredpalg by just replacing \greedyshortcut. \Cref{thm:dp-hopset-relation} therefore holds universally as the connection between exact hopsets and the DP-ASRQ problem.

\subsection{Discussion on All Pairs Shortest Distances}
\label{subsec:discuss-apsd}
The problems of APSD and ASRQ share many similarities. Given a graph with public topology and private weights, the APSD problem aims to output weighted shortest distances under DP, and minimize the additive error, which is again the $\ell_\infty$ norm of the difference between the true distances and the perturbed ones. We define the problem formally.

\begin{definition}[\textit{Differentially Private APSD~\cite{sealfon2016shortest}}]
    \normalfont Let $w, w':E\rightarrow R^{\geq 0}$ be weight functions, and $\mathcal{A}$ be an algorithm taking a graph $G = (V,E)$ and $w$ as input. The algorithm $A$ is $(\varepsilon, \delta)$-differentially private on $G$ if for any neighboring weights $w \sim w'$ and all sets of possible output $\mathcal{C}$, we have: $\Pr[\mathcal{A}(G, w)\in \mathcal{C}] \leq e^{\varepsilon}\cdot \Pr[\mathcal{A}(G, w') \in \mathcal{C} ]+\delta.$
\end{definition}

One may wonder why we cannot obtain a similar connection as \Cref{thm:dp-hopset-relation} between the DP-APSD problem and low-sensitivity hopsets. There is a crucial yet subtle difference between two problem that falls into the privacy notion: the weights that determine the shortest paths are private in the APSD problem, but are public knowledge in the ASRQ problem. In fact, the latter does not concern weights but other attributes irrelevant to the shortest paths, which are assumed to be given to algorithms. In other words, algorithms for DP-APSD cannot use the true shortest paths as a white-box however the other is fine.

The discussion above implies that we cannot use \greedyshortcut to construct a low-sensitivity hopset because the ground-truth shortest paths are used therein. The currently state-of-the-art algorithms~\cite{ghazi2022differentially} use the folklore hopset construction by sampling a set of vertices $S$ uniformly at random. This procedure does not involve using the shortest paths.

On the other hand, this subtlety does not refrain from having another low-sensitivity hopset construction that does not involve the shortest paths (or edge weights). If we do have an algorithm as such, which outputs a $(\beta,0)$-hopset with $\inorm{\esens}$ sensitivity, a similar theorem as \Cref{thm:pure-dp} still holds for the DP-APSD problem. Knowing that the current best upper bound for the $\epsilon$-DP setting is $\Ot(n^{2/3})$, if we have $\beta \cdot \inorm{\esens} = o(n^{4/3})$, we improve the previous best $\epsilon$-DP result!

\end{document}